\documentclass[10pt,aps,pra,twocolumn,showpacs,preprintnumbers,nofootinbib]{revtex4-1}

\usepackage{graphicx}  
\usepackage{tikz}
\usepackage{pgfplots}

\pgfplotsset{compat=1.18}
\usepgfplotslibrary{colormaps}

\usepackage[caption=false]{subfig}
\newcommand\subfigref[1]{\protect\subref{#1}}

\usepackage{multirow}

\usepackage{fancyhdr}
\usepackage{longtable}
\usepackage{parskip}
\usepackage[T1]{fontenc}
\usepackage{dcolumn}   
\usepackage{comment}
\usepackage{hyperref}
\usepackage{tikz}
\usepackage{tikz-3dplot}
\usepackage{pgfplots}

\usepackage{bm}        
\usepackage{amsfonts}  
\usepackage{amsmath}   
\usepackage{amssymb}   
\usepackage{stmaryrd}
\usepackage[toc,page]{appendix}
\usepackage{algpseudocode}
\usepackage{algorithm}
\usepackage{xcolor} 
\usepackage{minted}
\usepackage{hyphenat}
\usepackage{glossaries}
\usepackage{cleveref}
\usepackage{amsthm}

\newcommand{\bra}[1]{\left\langle #1 \right|}

\newcommand{\pwisein}{\left\{ \begin{array}{ll}}
\newcommand{\pwiseout}{\end{array}\right.}
\newcommand{\ket}[1]{\left| #1 \right\rangle}

\renewcommand{\det}[1]{\mathrm{det}\left( #1 \right)}

\newtheorem{theorem}{Theorem}
\newtheorem{lemma}[theorem]{Lemma}
\newtheorem{proposition}[theorem]{Proposition}

\newacronym[shortplural={MPS}]{MPS}{MPS}{matrix\hyp product state}
\newacronym{MPO}{MPO}{matrix-product operator}
\newacronym{SVD}{SVD}{singular\hyp value decomposition}
\newacronym{QCS}{QCS}{quantum\hyp computer simulator}
\newacronym{FSM}{FSM}{finite\hyp state machine}
\newacronym{ACA}{ACA}{adaptive cross approximation}
\newacronym{1D}{1D}{one\hyp dimensional}
\newacronym{QC}{QC}{quantum computer}
\newacronym{CDW}{CDW}{charge\hyp density wave}
\newacronym{SDW}{SDW}{spin\hyp density wave}
\newacronym{ARPES}{ARPES}{angle-resolved photoemission spectroscopy}
\newacronym{OBC}{OBC}{open\hyp boundary conditions}
\newacronym{PBC}{PBC}{periodic\hyp boundary conditions}
\newacronym{TEBD}{TEBD}{time\hyp evolution block-decimation}
\newacronym{TDVP}{TDVP}{time\hyp dependent variational principle}
\newacronym{iff}{iff}{if and only if}
\newacronym{DFT}{DFT}{density\hyp functional theory}
\newacronym{DMFT}{DMFT}{dynamical mean\hyp field theory}
\newacronym{DMRG}{DMRG}{density\hyp matrix renormalization group}
\newacronym{1DMRG}{1DMRG}{single-site density\hyp matrix renormalization group}
\newacronym{2DMRG}{2DMRG}{two-site density\hyp matrix renormalization group}
\newacronym{DMRG3S}{DMRG3S}{strictly single-site density\hyp matrix renormalization group}
\newacronym{iDMRG}{iDMRG}{inifinite\hyp size density\hyp matrix renormalization group}
\newacronym{tDMRG}{tDMRG}{time\hyp dependent density\hyp matrix renormalization group}
\newacronym{QMC}{QMC}{quantum Monte Carlo}
\newacronym{AIM}{AIM}{Anderson impurity model}
\newacronym{SIAM}{SIAM}{single impurity Anderson model}
\newacronym{LDA}{LDA}{local\hyp density approximation}
\newacronym{LBNL}{LBNL}{Lawrence Berkeley National Laboratory}
\newacronym{VQE}{VQE}{variational\hyp quantum eigensolver}
\newacronym{ED}{ED}{exact diagonalization}
\newacronym{QPT}{QPT}{quantum phase transition}
\newacronym{QCP}{QCP}{quantum critical point}
\newacronym{ETH}{ETH}{eigenstate thermalization hypothesis}
\newacronym{EHM}{EHM}{extended Hubbard model}
\newacronym{AKLT}{AKLT}{Affleck\hyp Lieb\hyp Kennedy\hyp Tasaki}
\newglossaryentry{QR}{name={QR},description={QR decomposition}}
\newacronym{TNS}{TNS}{tensor\hyp network state}
\newacronym{SM}{SM}{supplemental material}
\newacronym{NOO}{NOO}{natural orbital occupation}
\newacronym{NO}{NO}{natural orbital}
\newacronym{LRO}{LRO}{long\hyp range order}
\newacronym{qLRO}{qLRO}{quasi\hyp long\hyp range order}
\newacronym{SC}{SC}{Superconductivity}
\newacronym{VBGS}{VBGS}{valence bond ground-state}
\newacronym{PEPS}{PEPS}{projected entangled pair\hyp states}
\newacronym{ALS}{ALS}{alternating least squares}
\newacronym{BdG}{BdG}{Bogoljubov de-Gennes}
\newacronym{TFIM}{TFI}{transverse field Ising model}
\newacronym{PP}{PP}{projected purification}
\newacronym{BEC}{BEC}{Bose\hyp Einstein condensate}
\newacronym{JWT}{JWT}{Jordan\hyp Wigner transformation}
\newacronym{NISQ}{NISQ}{noisy intermediate scale quantum}
\newacronym{NN}{NN}{nearest\hyp neighbor}
\newacronym{NNN}{NNN}{next\hyp nearest\hyp neighbor}
\newacronym{SPDM}{SPDM}{single\hyp particle density matrix} 
\newacronym{PDF}{PDF}{probability density function}
\newacronym{HCB}{HCB}{hardcore bosons}
\newacronym{SF}{SF}{spinless fermions}
\newacronym{AME}{AME}{absolutely maximally entangled}%

\begin{document}

\title{Ensembles of random quantum states tunable from volume law to area law}

\author{Héloïse Albot, Sebastian Paeckel}

\affiliation {\it Department of Physics, Arnold Sommerfeld Center for Theoretical Physics (ASC),
Munich Center for Quantum Science and Technology (MCQST),
Ludwig-Maximilians-Universität München, 80333 München, Germany }

\date{April 16, 2026}

\begin{abstract}  
A standard approach to generate random pure quantum states relies on sampling from the Haar measure. 
However, the entanglement properties of such states present a fundamental challenge for their general applicability.
Here, we introduce the $\sigma$\hyp ensembles -- a family of random quantum states with only a single control parameter.
Crucially, these states are designed such that they can be tuned between volume\hyp law and area\hyp law behavior, which has been a major obstacle thus far.
We construct representatives of this ensemble by imposing a probability distribution on the eigenvalues of the successive subsystems, and subsequently reconstructing a compatible global state using the matrix product state (MPS) formalism.
Due to their area\hyp law entanglement, our approach circumvents the intractability of Haar\hyp random pure states in classical simulations of quantum systems and is more representative of typical Hamiltonian ground states.

\end{abstract}


\maketitle 

\section{Introduction}
The discovery of the intimate connection between the physical properties of quantum many\hyp body systems on the one hand, and their entanglement properties on the other, has been one of the key driving forces for the vast algorithmic improvements in the past two decades~\cite{Vidal2004,Eisert2006,Hastings2007,Schuch2008a,Schuch2008a,Eisert_2010}.
Most notably, relations between the classical simulability of strongly correlated systems, and the scaling of the bipartite entanglement entropy paved the way for powerful numerical tools for a broad class of quantum systems.
In particular, the so-called area law states -- states whose bipartite entanglement entropy scales with the boundary size rather than the volume of the enclosed subsystem -- turned out to be the computationally most relevant class~\cite{Schollwock_2011}.
Given such a classification, one main goal for the development of quantum algorithms for quantum computers is to identify and characterize the entanglement properties of quantum circuits.
Clearly, quantum supremacy can only be achieved if classical simulability is impossible, which then can be traced back to the question, how correlations, and thus entanglement, evolves throughout the circuit.
Here, a widely used model to analyze quantum circuits is using layers of random unitaries~\cite{Fisher2023,schuster2025}.
Crucially, such circuits are known to generate Haar\hyp random states as output states.
However, the ensemble of Haar\hyp random states is known to exhibit a volume law~\cite{Page_1993, Foong_1994, Sen_1996}, while practically relevant quantum circuits should generate outputs that fall into the class of area\hyp law states in nearly all known cases of physical relevance \cite{Hastings_2007, Wolf_2008, Aharonov_2011, Kull_2019, Shor_1997, Grover_1996, Boixo2014}.
This raises the fundamental question of how to properly capture both the intermediate and final states occurring in such circuits.
Moreover, recent results~\cite{Zhang_2022, Li_2023} seem to indicate that on noisy quantum devices, circuits which expected to produce volume\hyp law states are actually generating area\hyp law states.
In order to characterize such circuits, the question emerges whether there exist other types of ensembles of random states and how to generate representatives.
In this work, we address this open problem by introducing a construction scheme for random states, which allows to precisely control the entanglement properties of the generated states.
This problem is particularly challenging because uniformly drawn state coefficients in general generate volume\hyp law states~\cite{Bianchi2022}.
A further complication arises from the fact that area\hyp law states form a set of measure zero in the total Hilbert space~\cite{Eisert_2010}.
Here, we demonstrate via explicit construction that it is indeed possible to sample area\hyp law states, and by tuning only a single control parameter, the corresponding ensemble can be transitioned between area\hyp law and volume\hyp law states.
Thereby, we also introduce a classification of random state ensembles based on the entanglement content.
This also allows to relate the drawn random states to the notion of classical simulability, and we provide an explicit construction scheme in terms of~\glspl{MPS}.
\section{How to sample random states from a set of measure zero\label{sec:sampling:general}}
The most natural approach to generate pure, random quantum states is to use a generalization of the uniform distribution to the case of high\hyp dimensional Hilbert spaces.
For this approach, the idea is to start from one\hyp dimensional representations of the $U(1)$ group, introducing a measure to sample uniformly from the complex unit circle.
The generalization to finite\hyp dimensional representations is straightforward by preserving the group properties in the $N$\hyp dimensional case, and known as the Haar random ensemble~\cite{Haar_1933, Weil_1953, Cartan_1940}.
While this measure is well characterized~\cite{lubkin_1978, Facchi_2008, majumdar2010} and studied in various contexts~\cite{Bengtsson_Zyczkowski_2006, Cappellini_2006, Page_1993, Majumdar_2008, Zyczkowski_2001, Giraud_2007, Znidaric_2007, Chen_2010, loio2024}, it generates uncommon pure states since they are always volume\hyp law states~\cite{Dubey2012}.
Here, by volume law we refer to the linear scaling of the entanglement entropy when varying the subsystem size (see~\cref{eq:entanglement-entropy:scalings}).
To be precise, consider a tensor\hyp product Hilbert space describing $L$ qubits.
The reduced density operator $\hat \rho_A$ is obtained by tracing out the subsystem $B$
\begin{equation}
    \hat \rho_A = \operatorname{Tr}_B(\hat \rho) \; .
\end{equation}
In this work, we consider only bipartitions such that we always have $\lvert A \rvert + \lvert B \rvert = L$, where $\lvert X \rvert$ denotes the number of qubits contained in subsystem $X$.
Furthermore, we use the von Neumann entropy as measure for the entanglement entropy given for a certain bipartition by
\begin{equation}
    S_A \equiv -\operatorname{Tr} \hat \rho_A \log \hat \rho_A \; .
    \label{eq:def_von_neumann_entropy}
\end{equation}
The physically most relevant scaling behaviors of $S_A$ can then be classified according to
\begin{equation}
    S_A \propto
    \begin{cases}
        \lvert A \rvert &\text{volume\hyp law state}\\
        \log \lvert A \rvert &\text{critical state}\\
        \lvert \partial A \rvert &\text{area\hyp law state}
    \end{cases}\label{eq:entanglement-entropy:scalings}
\end{equation}
where $\partial A$ denotes the set of boundary sites of the subsystem $A$.
Note that according to the Haar measure, both critical and area\hyp law states are sets with measure zero.
An intuition why this is the case can be gained by computing the expectation value of the entanglement entropy $\mathbb E[S_A]$ for $\lvert A\rvert < \lvert B\rvert$~\cite{Page_1993,Foong_1994, Sen_1996}
\begin{equation}
    \mathbb E[S_A] = \sum_{k=d^{\lvert B\rvert}+1}^{d^L} \frac{1}{k} - \frac{d^{\lvert B \rvert} - 1}{2d^{\lvert B\rvert}}\;,
    \label{eq:page_theorem}
\end{equation}
where $d$ denotes the local Hilbert space dimension, i.e., $d=2$ in the case of qubits.
A simple estimation then shows that for $1 \ll \lvert A\rvert$
\begin{equation}
    \mathbb E[S_A] \sim \lvert A \rvert \log d\;,
\end{equation}
and thus typical states drawn according to the Haar measure are asymptotically volume\hyp law states.
Let us now illustrate how this limitation can be overcome, considering the case $\lvert A \rvert \leq \lvert B \rvert$.
For that purpose, note that for any bipartition into subsystems $A,B$, the maximally entangled state has density operator
\begin{equation}
    \hat \rho_A = \frac{1}{2^{\lvert A \rvert }} \hat{\mathbf 1}_{2^{\lvert A\rvert} \times 2^{\vert A
    \rvert}} \;. \label{eq:density-operator:volume-law}
\end{equation}
On the other hand, the minimal entanglement entropy is reached by a product state with density operator
\begin{equation}
    \hat \rho_A = \ket{\psi_A}\!\!\bra{\psi_A}\; , \label{eq:density-operator:area-law}
\end{equation}
on each partition, and $\ket{\psi_A} \in \mathcal H_A$ is a pure state in the Hilbert space spanned by the qubits in the subsystem $A$.
We now give a geometric meaning to~\cref{eq:density-operator:volume-law,eq:density-operator:area-law} by assigning the sequence of eigenvalues $\left(\lambda_1,\ldots,\lambda_n\right)$ of a density operator $\hat \rho_A$ to a point on the positive orthant of the $n=2^{\lvert A \rvert}$\hyp dimensional unit sphere
\begin{equation}
    \operatorname{spec} \hat \rho_A = \vec{x}\circ \vec{x} , \quad \text{ for }\vec{x} \in \mathcal{S}^n_+  \;, \label{eq:density-operator:sampling-scheme}
\end{equation}
where $\mathcal{S}^n_+= \{ \vec x \in \mathbb{R}^{n}\; \vert\; \lVert \vec x \rVert_2 = 1, \; \; x_i \geq 0 \; \forall i \in \llbracket 0, n+1 \rrbracket\}$ and $\circ$ denotes the element-wise product.
Note that in doing so, we choose an (arbitrary) assignment between the eigenvalues and the corresponding eigenvectors of $\hat \rho_A$.
To illustrate this point, note that for two density operators $\hat \rho_A, \hat \sigma_A$ with ${\operatorname{spec} \hat \rho_A} = (\lambda_1, \lambda_2)$ and ${\operatorname{spec} \hat \sigma_A} = (\lambda_2, \lambda_1)$, we cannot assume $\hat\rho_A = \hat\sigma_A$ because the corresponding eigenvectors can differ.
Keeping this subtlety in mind, a density operator, as shown in~\cref{eq:density-operator:volume-law}, is now completely specified by the unique point on the positive orthant of the unit sphere, which has coordinates $x_i = 1/\sqrt{2^{\lvert A\rvert}}$ for all $i\in\left\{ 1,\ldots 2^{\lvert A \rvert} \right\}$.
In the case of a pure state, as shown in~\cref{eq:density-operator:area-law}, any point described by a vector $\vec{x} $ parallel to a coordinate axis $\vec e_i$ yields a valid representation.
Since a pure quantum state is fully characterized by all its reduced density operators~\cite{Xin_2017}, the previous observations suggest to introduce an ensemble by uniformly sampling the eigenvalues of the reduced density operators in each bipartition from a subset $\mathcal L \subset \mathcal S^n$, as well as random basis states corresponding to the eigenbasis of the reduced density operators \footnote{Here, we will define $n-1$ reduced density operators over the $2^{L-1}-1$ possible bipartitions. Consequently, our constraints are less restrictive and do not fully characterize the states, leaving some degrees of freedom.}.
We now use this geometric picture to motivate our method to sample area\hyp law states.
The key insight is that upon increasing the dimensionality and assuming a uniform sampling $\mathcal L = \mathcal{S}^{n}_+$, the probability to encounter very small coordinates $\lambda_i$ increases quickly due to the normalization constraint $\sum_{i=1}^{2^{\vert A \rvert}} \lambda_i^2 = 1$.
The resulting points on the unit sphere will therefore have only very few, non\hyp vanishing elements, i.e., their coordinate vectors $\vec \lambda$ are nearly parallel to very few randomly chosen coordinate axis with very high probability.
Thus, owing to the exponential increase of the Hilbert space of the subsystem $A$, a uniform sampling over $\mathcal{S}^n_+$ will almost certainly generate area\hyp law states.
In the next sections, we discuss this idea in more detail and investigate the properties of the suggested ensemble when varying the sampling scheme.
\subsection{$n$-sphere sampling scheme\label{sec:sampling-scheme:eigenvalues}}
\begin{figure}[ht]
    \centering

\begin{tikzpicture}
    \begin{axis}[
        xlabel={Index $i$},
        ylabel={$\bar{\lambda_i}(\sigma)$},
        xmin=-4, xmax=130,  
        ymin=0, ymax=80,
        ymode=log,
        log origin=infty,
        grid=both,
        minor tick num=1,
        legend pos=north east,
        mark=o,
    ]

   \addplot[
    mark=*, 
    mark size=1.5pt, 
    color=blue!70!black, 
   ] table [x index=0, y index=1] {Data/lambda_plot_reglin_whole_orthant_n128_logscale.txt};

    \end{axis}
\end{tikzpicture}
    \caption{\textbf{Mean of the $i$-th largest eigenvalue generated by taking the square of Cartesian coordinates of a uniformly random point on the $n$-sphere with $\mathbf{n=128}$.} We average the decreasingly-ordered eigenvalues over $1000$ simulations and use a logarithmic scale for the y-axis.}
    \label{fig:eigenvalues:uniform-sampling}
\end{figure}

\begin{figure}[ht]
    \centering
    \includegraphics[width=1\linewidth]{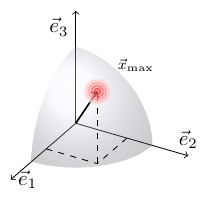}
    \caption{\textbf{Positive orthant of the $n$-sphere from which the random points are sampled.} The squares of their Cartesian coordinates constitute the eigenvalues of the corresponding density matrix. Here, we illustrate the case for $n=3$. When sampling random quantum states that obey an area law, points are drawn uniformly from the entire positive orthant (gray area). When sampling states transitioning between a volume law and an area law, the points follow a Gaussian distribution (red circles). This distribution is centered at the tip of the arrow, corresponding to the coordinates $(x,y,z)=(\frac{1}{\sqrt{n}},\frac{1}{\sqrt{n}},\frac{1}{\sqrt{n}})$, with a standard deviation parameterized by $\sigma$.}
    \label{fig:eigenvalues:gaussian-sampling}
\end{figure}
Let us begin by considering without loss of generality a bipartition where $\lvert A \rvert \leq \lvert B \rvert$ and investigate the case of uniform sampling over $\mathcal L = \mathcal{S}^{n}_+$.
The practical implementation of this sampling procedure is detailed in~\cref{annex:generating_the_eigenvalues_uniform}.
In~\cref{fig:eigenvalues:uniform-sampling}, we display the means of the eigenvalues obtained when averaging the decrasingly\hyp ordered eigenvalues drawn for $1000$ realizations, assuming a Hilbert space dimension $n = 128$ for the subsystem $A$.
We clearly observe an exponential decay, indicating that on average we find indeed an area\hyp law state.
This numerical observation is supported by the explicit calculation of the expectation values of the eigenvalues in~\cref{annex:computation_of_exponential_decay}.
Complementary to this evidence, we calculate the expectation value of the von Neumann entropy as a function of the subsystem's Hilbert space dimension.
For this purpose, we evaluate the expectation value of the von Neumann entropy in~\cref{eq:def_von_neumann_entropy}, which can be expressed as $S_A(\vec\lambda) = -\sum\limits_{i=1}^n \lambda_i \log \lambda_i$, for samples of density operator eigenvalues $\vec\lambda$ uniformly drawn from $\mathcal S^n_+$ with $n=2^{\lvert A \rvert}$
\begin{equation}
    \mathbb E[S_A] = \frac{1}{\lvert S^n_+\rvert}\int_{S^n_+} S_A(\vec\lambda) d^n\vec\lambda \; .
\end{equation}
With the full details of the calculation given in~\cref{appendix:ananlytical_proof_average_von_neumann_entropy}, we obtain
\begin{equation}
    \mathbb{E}[S_A] =  \left(\ln(2)-\frac{1}{2} \right) \left(4- \frac{1}{2^{2^{\vert A \rvert}-3}} \right) \; .
\label{eq:density-operator:our_method}
\end{equation}
Crucially, and in contrast to the Haar measure, this expectation value is not diverging as $n\rightarrow \infty$.
Instead, we arrive at
\begin{equation}
    \lim_{\lvert A \rvert \rightarrow \infty}\mathbb{E}[ S_A ] = 4\ln 2 - 2 \;. \label{eq:density-operator:area-law:exact}
\end{equation}
We note that in this expression there is no dependency on the size of the subsystem $A$, hence $\mathbb{E}[S_A]$ can be associated to area\hyp law states if the bipartition is chosen such that $\lvert \partial A \rvert = const.$ for any possible subsystem $A$.
As a consequence, for one\hyp dimensional systems -- or for higher\hyp dimensional systems mapped into a chain -- the reduced density operators $\hat \rho_A$ describe area\hyp law states on average when the eigenvalues are sampled uniformly from the $n$-sphere.
This holds for bipartitions that decompose the system into one\hyp dimensional chains.
~\Cref{eq:density-operator:volume-law} allows to define a second specific choice of the sample space $\mathcal L$, here with the goal to generate volume\hyp law states.
In fact, choosing $\mathcal L = \left\{\vec x \in \mathbb R^{n+1} \, \vert \, x_i = 1/\sqrt{n}\right\}$, the resulting density operator has von Neumann entropy $S_n = \log_2(n)$ and, since $n = 2^{\lvert A \rvert}$, this corresponds to the specific volume\hyp law state with maximal entanglement entropy.
We now propose to interpolate between these two limiting cases by introducing a Gaussian probability distribution function for the spherical coordinates $\vec \theta$ of the point $\vec X$ on the $n$-sphere 
\begin{equation}
    p_\sigma(\vec \theta) = \frac{1}{\left(2\pi\sigma^2\right)^{n/2}} \mathrm{exp}\left(-\frac{(\Delta \vec \theta)^\mathrm{T} \mathbf G \Delta \vec \theta}{2}\right) \;, \label{eq:eigenvalues:gaussian-sampling}
\end{equation}
where $\Delta \vec \theta = \vec \theta - \vec \theta_\mathrm{max}$ and $\vec \theta_\mathrm{max}$ is the vector of spherical coordinates of $\vec x_\mathrm{max}=\left(1/\sqrt{n}, \ldots, 1/\sqrt{n} \right)^\mathrm{T}$, 
$\mathbf G= \hat{\mathbf 1}_{n \times n} / \sigma^2 $ is the precision matrix and $\sigma > 0$ characterizes the width of the Gaussian distribution.
Once the spherical coordinates are defined, it is easy to deduce the corresponding Cartesian coordinates to obtain the eigenvalues using~\cref{eq:density-operator:sampling-scheme}.
This sampling scheme is depicted in~\cref{fig:eigenvalues:gaussian-sampling} and the procedure is detailed in \cref{annex:generating_the_eigenvalues_gaussian}.
We immediately find that $\sigma\rightarrow 0$ corresponds to the maximally entangled case.
Moreover, in the limit $\sigma\rightarrow\infty$ one may hypothesize that the uniform sampling is recovered.
At this point it should be noted that using a Gaussian distribution for the spherical coordinates, the sampled points are no longer restricted to the positive orthant. 
As a consequence, a set of eigenvalues $\vec{\lambda}$ will have $2^n$ corresponding points on the $n$-sphere, one per orthant. 
Nevertheless, in the limit $\sigma\rightarrow\infty$, one still recovers a uniform sampling, and therefore an ensemble characterized by \cref{eq:density-operator:our_method}.
\subsection{$\sigma$\hyp Ensemble of random pure states\label{sec:sampling-scheme:states}} 
We now generalize the gaussian sampling introduced in~\cref{sec:sampling-scheme:eigenvalues} to generate global random pure states whose entanglement content can be controlled by the width $\sigma$ of the Gaussian distribution~\cref{eq:eigenvalues:gaussian-sampling}.
Given a set of density matrices, the problem of finding a compatible global pure state is known as the quantum marginal problem. This problem is generally computationally intractable, even for a quantum computer~\cite{QMA_Liu_2007}. 
While compatibility conditions have been established for certain configurations~\cite{Lieb_1973, Higuchi_2003, bravyi2003requirementscompatibilitylocalmultipartite, Klyachko_2004, Klyachko_2006, Christandl_2017, Makhina2023}, no general solution exists. 
Resolving this problem is beyond the scope of this paper.
Here, we present a method to construct a compatible global pure state in the special case where we have an ensemble of the singular value diagonals $S^{[1]}, \ldots S^{[L-1]}$, referred to as a \emph{$\sigma$\hyp ensemble}.
This state is represented as a left\hyp canonical~\gls{MPS}~\cite{Schollwock_2011}: 
\begin{equation}
    \ket{\psi} = \sum_{\sigma_1,\ldots,\sigma_L} A^{\sigma_1} \ldots A^{\sigma_L}\ket{\sigma_1,\ldots,\sigma_L} \;,
\end{equation}
where $A^{\sigma_l} \in \mathbb C^{m_{l-1} \times m_l} $ are left\hyp normalized matrices, i.e., $\sum_{\sigma_l} \left(A^{\sigma_l} \right)^{\dagger} A^{\sigma_l} = \hat{\mathbf 1}_{m_l \times m_l}$, with bond dimension $m_l\in\mathbb N$.
Here, $\sigma_l\in\left\{0,1\right\}$ denotes the basis states of the $l$-th local qubit.
Suppose that the left\hyp normalized matrices $\{A_{\sigma_l}\}_{\sigma_l}$ have been constructed up to site $0\leq l \leq L -1$ such as for $1\leq i \leq l$, the diagonal values of $ (S^{[i]})^2$ are the singular values of $\operatorname{Tr}_{\sigma_{i+1},\ldots, \sigma_L}(\ket{\psi}\bra{\psi})$ where $\operatorname{Tr}_{\sigma_{i+1},\ldots, \sigma_{L} }$ denotes the partial trace operator over the sites $i+1, \ldots, L$. 
We now proceed to construct the set $\{A_{\sigma_{l+1}}\}_{\sigma_{l+1}}$.
As shown in~\cite{Schollwock_2011}, the state $\ket{\psi}$ can be expressed as:
\begin{equation}
    \ket{\psi} = \sum_{\sigma_1,\ldots,\sigma_L} A^{\sigma_1} \ldots A^{\sigma_l} S^{[l]} M^{\sigma_{l+1}} M^{\sigma_{l+2}} \ldots M^{\sigma_{L}} \ket{\sigma_1,\ldots,\sigma_L} \;
    \label{eq:mixed-canonical-MPS} 
\end{equation}
where $S^{[l]}$ is the diagonal matrix composed of the singular values on the bond $\left( l,l+1\right)$. 
If $l=0$, $S^{[0]}=\begin{pmatrix} 1\end{pmatrix}$, otherwise $S^{[l]}$ is determined via the procedure in~\cref{sec:sampling-scheme:eigenvalues}.
The matrices $\{M^{\sigma_{l+1}}\}_{\sigma_{l+1}}$ must satisfy the following~\gls{SVD} relation:
\begin{equation}
    \Sigma^{[l]} \Omega^{[l]}M^{[l+1]} = U^{[l+1]}S^{[l+1]}W^{[l+1]}
  \label{eq:svd:Sigma_Omega_B}
\end{equation}
where $M^{[l+1]}$ is defined as
\begin{equation}
    M^{[l+1]}
    =
    \left(
        \begin{array}{c}
            M^{\sigma_{l+1}=0} \\ 
            \vdots \\
            M^{\sigma_{l+1}=d-1}
        \end{array}
    \right) \; , 
\end{equation}
$\Sigma^{[l]}$ is a diagonal matrix:
\begin{equation}
    \Sigma^{[l]}
    =
    \setlength{\arraycolsep}{0pt}
  \setlength{\delimitershortfall}{0pt}
  \begin{pmatrix}
    \,\fbox{$S^{[l]}$} & 0 & \ldots & 0  \,  \\
    \,0 & \fbox{$S^{[l]}$} & \ddots & \vdots \, \\
    \,\vdots  & \ddots & \ddots & 0 \, \\
    \,0 & \ldots & 0 & \fbox{$S^{[l]}$}\, \\
  \end{pmatrix} \;,
\end{equation}
and $\Omega^{[l]}$ is a block\hyp diagonal matrix:
\begin{equation}
    \Omega^{[l]}
    =
    \setlength{\arraycolsep}{0pt}
  \setlength{\delimitershortfall}{0pt}
  \begin{pmatrix}
    \,\fbox{$W^{[l]}$} & 0 & \ldots & 0  \,  \\
    \,0 & \fbox{$W^{[l]}$} & \ddots & \vdots \, \\
    \,\vdots  & \ddots & \ddots & 0 \, \\
    \,0 & \ldots & 0 & \fbox{$W^{[l]}$}\, \\
  \end{pmatrix} \;
\end{equation}
where $W^{[l]}$ was defined in the previous iteration, or $W^{[l]}=\begin{pmatrix} 1\end{pmatrix}$ if $l=0$.
Here, $U^{[l+1]}$ is left\hyp unitary and $W^{[l+1]}$ is right\hyp unitary and square, hence unitary. 
We express $M^{[l+1]}$ via its reduced QR decomposition: 
\begin{equation}
    M^{[l+1]} = Q^{[l+1]} R^{[l+1]}
    \label{eq:qr_M}
\end{equation}
where $Q^{[l+1]\dagger} Q^{[l+1]} = \hat{\mathbf 1}_{dm_l \times dm_l}$ and $R^{[l+1]}$ is upper triangular.
We then set
\begin{equation}
    X^{[l+1]}=\Sigma^{[l]}\Omega^{[l]} Q^{[l+1]}
\end{equation}
and perform its~\gls{SVD}:
\begin{equation}
    X^{[l+1]}= Y^{[l+1]} \Delta^{[l+1]} Z^{[l+1]}
    \label{eq:svd_X}
\end{equation}
where $Y^{[l+1]\dagger}Y^{[l+1]}=\hat{\mathbf 1}_{m_{l+1} \times m_{l+1}}$, $Z^{[l+1]}Z^{[l+1]\dagger}=\hat{\mathbf 1}_{m_{l+1} \times m_{l+1}}$ and $\Delta^{[l+1]}$ is diagonal.
We have:
\begin{align*}
    & \Sigma^{[l]}\Omega^{[l]}M^{[l+1]}M^{[l+1]\dagger}\Omega^{[l]\dagger}\Sigma^{[l]} \\
    =& Y^{[l+1]} \Delta^{[l+1]} Z^{[l+1]} R^{[l+1]} R^{[l+1]\dagger }Z^{[l+1]\dagger }\Delta^{[l+1]}Y^{[l+1]\dagger} \;.
\end{align*}
From~\cref{eq:svd:Sigma_Omega_B}, we have simultaneously that:
\begin{equation}
    \Sigma^{[l]}\Omega^{[l]}M^{[l+1]}M^{[l+1]\dagger}\Omega^{[l]\dagger}\Sigma^{[l]} =U^{[l+1]}(S^{[l+1]})^2 U^{[l+1]\dagger} \; .
\end{equation}
We fix $Y^{[l+1]}$, a gauge degree of freedom of the~\gls{MPS}, such that $Y^{[l+1]}=U^{[l+1]}$.
This yields:
\begin{equation}
    \Delta^{[l+1]} Z^{[l+1]} R^{[l+1]} R^{[l+1]\dagger }Z^{[l+1]\dagger }\Delta^{[l+1]} = (S^{[l+1]})^2 
\end{equation}
and then:
\begin{equation}
    R^{[l+1]} R^{[l+1]\dagger } = Z^{[l+1]\dagger }(\Delta^{[l+1]})^{-1}(S^{[l+1]})^2 (\Delta^{[l+1]})^{-1}Z^{[l+1]}\;.
\end{equation}
At this point, we draw $Q^{[l+1]}$ randomly according to the Haar measure, compute $X^{[l+1]}$, and deduce $\Delta^{[l+1]}$ and $Z^{[l+1]}$ via its~\gls{SVD} shown in~\cref{eq:svd_X}.
We could then use the Cholesky decomposition to determine $R^{[l+1]}$ and finally obtain $M^{[l+1]}$ via~\cref{eq:qr_M}. 
However, it turns out that an eigenvalue decomposition of the right\hyp hand side is more stable because we can enforce positive semidefiniteness, which otherwise may be violated at the order of the square root of the numerical precision.
Once $\{M^{\sigma_{l+1}}\}_{\sigma_{l+1}}$ are determined, it is easy to deduce $\{A^{\sigma_{l+1}}\}_{\sigma_{l+1}}$ via an~\gls{SVD} \cite{Schollwock_dmrg}, and by construction the diagonal entries of $(S^{[l+1]})^2$ are the singular values of $\operatorname{Tr}_{\sigma_{l+2},\ldots, \sigma_{L}}(\ket{\psi}\bra{\psi})$.
This procedure is iterated until $L$ site tensors are constructed.
However, care must be taken.
Due to the quantum marginal problem, the obtained set of site matrices does not provide a global representation of an element of the $\sigma$\hyp ensemble in general.
We elaborate on this in more detail in~\cref{annex:algorithm:sweeping}, where we also introduce a sweeping procedure.
This procedure uses the site matrices as preconditioned input for an iterative update scheme and yields the~\gls{MPS} representation which exhibits the desired sampled Schmidt values $S^{[l]}$ in each bipartition of the chain.
This construction is fully determined by the standard deviation $\sigma$ of the Gaussian distribution.
Additionally, we introduce a second control parameter: the maximum allowed bond dimension $\chi$.
This can be incorporated by truncating the drawn eigenvalues and retaining at most $\chi$ column vectors.
As a result, we can control the computational complexity and thereby the simulability of the ensemble from which the quantum states are drawn.
\section{Physical properties of the $\sigma$\hyp ensemble\label{sec:characterization}}
\begin{figure}[ht!]
    \centering
    \subfloat[\label{fig:von_neumann_entropy_3D_surface_plot_average_over_20_samples_until_16_gaussian_disk_log}]{
    \includegraphics[width=.95\linewidth]{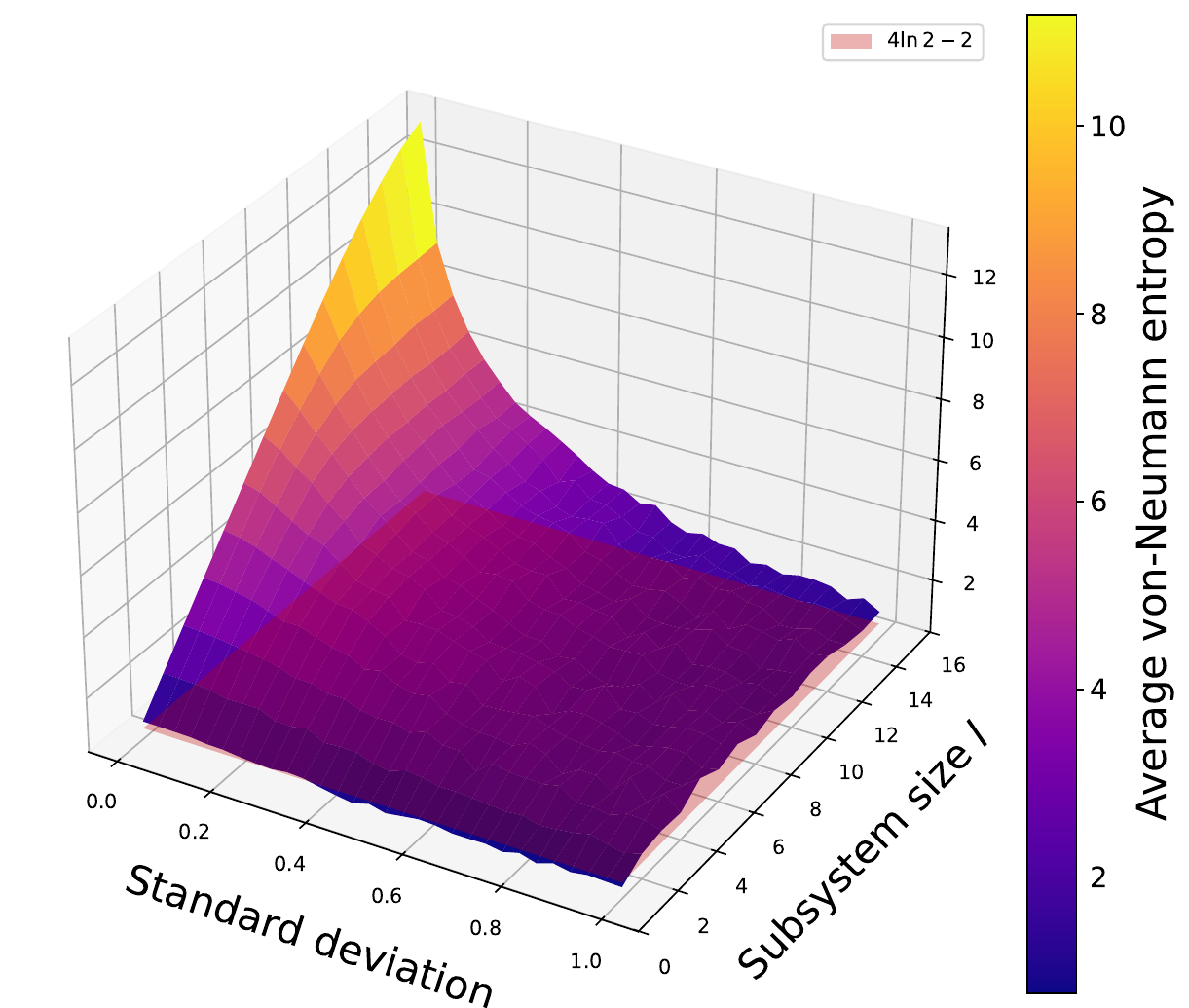}
    }\\
    \subfloat[\label{fig:bond_dimension_3D_surface_plot_average_over_10_samples_until_9_gaussian_disk_log}]{
        \includegraphics[width=.95\linewidth]{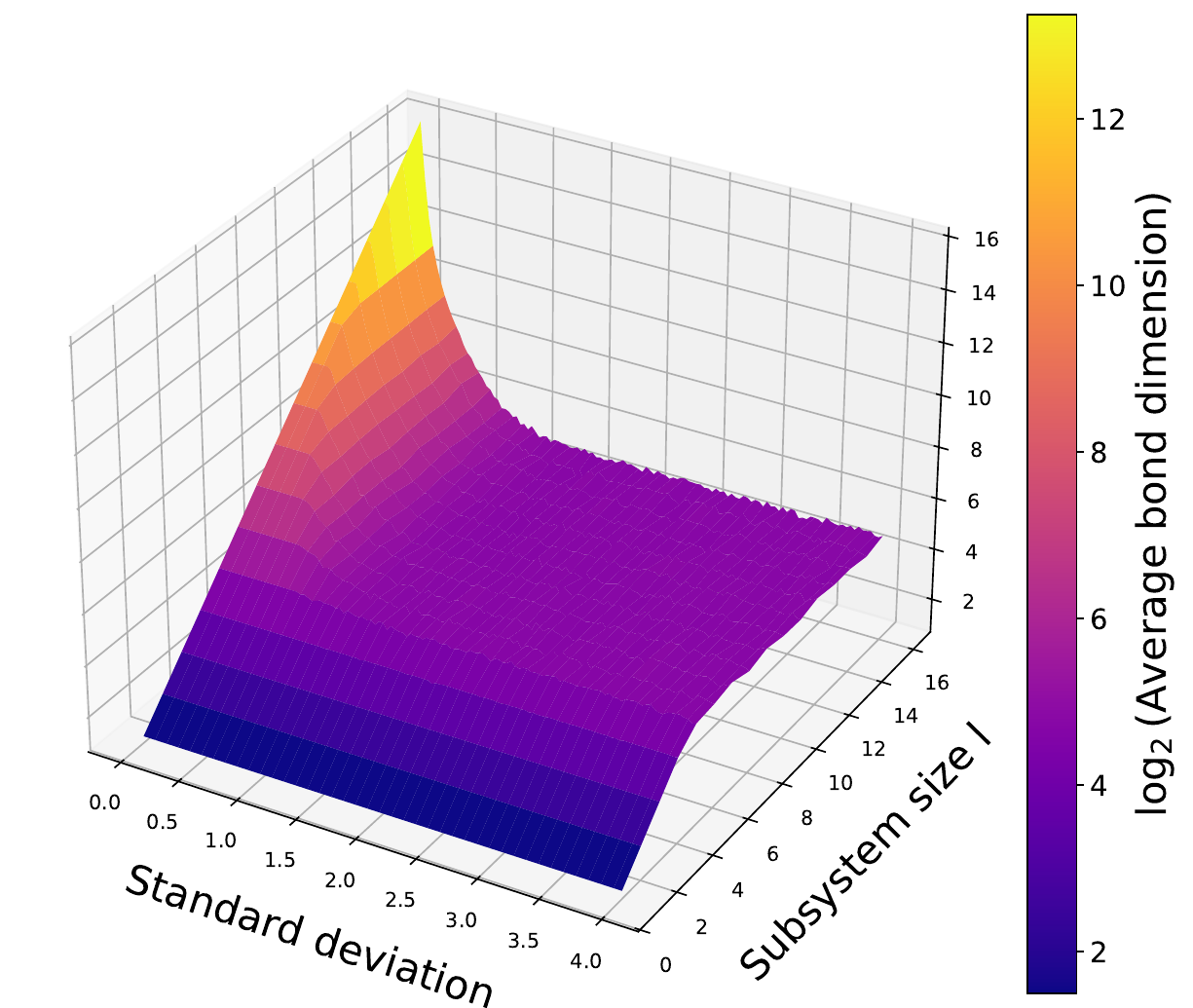}
    }
    \caption{
        \textbf{Mean von Neumann entropy~\subfigref{fig:von_neumann_entropy_3D_surface_plot_average_over_20_samples_until_16_gaussian_disk_log} and maximum bond dimension~\subfigref{fig:bond_dimension_3D_surface_plot_average_over_10_samples_until_9_gaussian_disk_log} of a random state obtained by sampling the eigenvalues from our Gaussian distribution, as a function of the subsystem size $l$ and the standard deviation $\sigma$.}
        The means were computed over $20$ samples, with a truncation threshold of $10^{-16}$ for the eigenvalues of the reduced density operator.
        Increasing the width $\sigma$ of the distribution, i.e., approaching uniform sampling of the eigenvalues from the unit sphere, reveals an area law, characterized by a subsystem\hyp size independent von Neumann entropy and a saturated maximum bond dimension.
        Conversely, the limit $\sigma \to 0$ yields a volume law, with an exponentially increasing maximum bond dimension.
    }
\end{figure}
As discussed previously, our method enables the generation of random pure quantum states with a controlled amount of entanglement.
In this section, we characterize the physical properties of the resulting ensemble of random states.
We first analyze the sampled eigenvalues to understand how tuning $\sigma$ interpolates between area\hyp law and volume\hyp law states, and how this affects the simulability, i.e., the computational complexity of representing and manipulating these states.
To that end, we focus on the bond dimension, which is a key quantity for tensor network simulation methods such as~\gls{MPS} since smaller bond dimensions allow for more efficient numerical representations and simulations.
Subsequently, we quantity the success of our algorithm in constructing a compatible global pure state from a given $\sigma$\hyp ensemble. 
Finally, we construct a phase diagram of the $\sigma$\hyp ensembles.
Note that in what follows, we always assume that the system is described by an underlying tensor\hyp product Hilbert space for which a one\hyp dimensional chain mapping has been chosen.
Furthermore, all analyses are performed for qubits with local Hilbert space dimension $d=2$.
\subsection{Entanglement entropy and simulability\label{sec:characterization:entanglement-and-simulability}}
From the general considerations in~\cref{sec:sampling:general}, we already know the two limiting cases of the eigenvalue distributions: (i) a maximally entangled state for $\sigma=0$, and (ii) an area\hyp law state with von Neumann entropy $\mathbb E[S_\infty] \sim 4\ln 2 - 2$ as $\sigma\to\infty$, where $S_\infty$ denotes the limit $\lim_{\lvert A \rvert \rightarrow \infty} S_A$.
To explore how these two limits are connected, we evaluate the mean von Neumann entropy of eigenvalue distributions tuning both the standard deviation $\sigma$ and the size $l$ of the smaller subsystem of the bipartition.
In~\cref{fig:von_neumann_entropy_3D_surface_plot_average_over_20_samples_until_16_gaussian_disk_log}, we show the resulting entanglement entropies for $\sigma\in(0,1]$.
For very small standard deviations $\sigma \approx 0$, the von Neumann entropy approaches its maximum value $S_n = \lvert l\rvert \log(2)$, clearly indicating a volume law.
Whereas as $\sigma$ increases, a rapid transition to an area\hyp law state occurs, evidenced by the saturation of the von Neumann entropy with increasing $l$.
Interestingly, for a fixed subsystem size $l$, the average von Neumann entropy decays towards its minimum value $\sim\mathbb E[S_\infty]$ with increasing $\sigma$.
This implies that, for any finite\hyp dimensional system, the length scale over which the sampled eigenvalue distributions generate volume\hyp law states decreases exponentially with increasing $\sigma$. 
Thus, in the infinite\hyp system limit, the corresponding quantum states are effectively sampled from the set of area\hyp law states. 
Moreover, we also plot the bond dimension as a function of the standard deviation and the system size in order to characterize the evolution of the simulability according to $\sigma$ and $n$.
Indeed, numerical methods such as the~\gls{DMRG} are known to be powerful tools for studying the low\hyp energy properties of one\hyp dimensional quantum systems, including spin chains, lattice models in condensed matter physics, and certain classes of strongly correlated systems~\cite{Schollwock_dmrg, Hallberg_2006}. 
However, their efficiency heavily depends on the bond dimension~\cite{Schollwock_2011}. 
With our method, for a Gaussian distribution centered around the maximally entangled state, the bond dimension, and thus the simulability, of our random states is determined by the standard deviation as can be seen in \cref{fig:bond_dimension_3D_surface_plot_average_over_10_samples_until_9_gaussian_disk_log}.
This highlights a key advantage of our sampling scheme: by appropriately choosing $\sigma$, we can generate random states with a desired bond dimension.
This figure also emphasizes that the range of $\sigma$'s for which the random states have a small bond dimension is far from marginal.
\subsection{Finding a compatible state \label{subsec:efficiency_algo}}
In this section, we evaluate the effectiveness of our method in reconstructing a compatible global pure state from a $\sigma$\hyp ensemble.
We recall that once the $L-1$ sets of eigenvalues are sampled, a compatible global pure state can be generated following the process described in~\cref{sec:sampling-scheme:states}. 
To quantify the performance of this method, we choose an error tolerance $\epsilon$ per state, and we compute the fraction of states whose error is smaller than~$\epsilon$.
We define this rate as the \emph{admission rate}:
\begin{equation}
    \mathcal{R}_{\epsilon} \equiv \lim_{N_\text{states}\to+\infty} \frac{1}{N_\text{states}} \sum_{s=1}^{N_\text{states}} \mathbf{1}(\text{error}((\vec \lambda_\text{target}, \vec \lambda_\text{GS})_s) < \epsilon) \; 
    \label{eq:admission_rate}
\end{equation}
where $\mathbf{1}(\cdot)$ denotes the indicator function.
The admission rate depends a priori on both the system size $L$ and the standard deviation $\sigma$ and is estimated numerically by sampling a large number of states $N_\text{states}$.
We choose $N_\text{states}=10,000$ for our numerical estimations.
The error of the constructed global pure state is quantified by comparing the target ensemble of eigenvalue sets $\vec \lambda^i_\text{target}$ with the ensemble of sets of eigenvalues of the $L-1$ bipartitions obtained with a single cut $\vec \lambda^i_\text{GS}$ from the global pure state:
\begin{equation}
    \text{error}((\vec \lambda_\text{target}, \vec \lambda_\text{GS}) )\equiv \frac{1}{L-1} \sum_{i=1}^{L-1} \rvert \rvert \vec \lambda^i_\text{target} -  \vec \lambda^i_\text{GS}  \lvert \lvert_2 \;.
\end{equation}
Here, $\vec \lambda^i_\text{target}$ denotes the vector formed from the set of eigenvalues sampled from a point on the $n$-sphere for the bipartition obtained with the cut between the site $i$ and $i+1$, and $\vec \lambda^i_\text{GS}$ is the corresponding vector formed from the set of eigenvalues from the global state constructed for the same bipartition.
This definition captures the average Euclidean (two\hyp norm) deviation between the target eigenvalues and those produced by the algorithm, averaged over all bipartitions.
In~\cref{fig:fraction_vs_L_for_different_chi}, we plot $\mathcal{R}_{10^{-4}}$ as a function of the number of lattices sites~$L$ for $\sigma=+\infty$. 
The results show that a finite fraction of samples achieves an acceptable error of $10^{-4}$. 
By increasing the number of samples, it seems therefore always possible to generate a random area\hyp law state for a given system size $L$.
At first glance, the exponential decay of the admission rate with $L$ may indicate that generating a random area\hyp law state requires running our algorithm $\mathcal{O}(e^{L})$ times. 
However, as we detail below, tuning the parameter $\sigma$ can significantly reduce the prefactor of this exponential decay.
In particular, there exists a regime in which one can efficiently sample random states that satisfy the area law.
\begin{figure}[H]
    \centering
    \begin{tikzpicture}
    \begin{axis}[
        xlabel={$L$},
        ylabel={$\mathcal{R}_{10^{-4}}$},
        xmin=8, xmax=14,
        ymin=0, ymax=0.025,
        xtick=data,
        grid=both,
        legend style={font=\small, xshift=-0.2cm, yshift=-0.2cm, anchor=north east},
        name=main plot, 
    ]

    \addplot[
    mark=*, 
    mark size=1.5pt, 
    color=green!60!black, 
    line width=1pt,   
    ] table [x index=0, y index=1] {Data/fraction_vs_L_data_for_different_chi.txt};

    \addlegendentry{$\chi_{\max} = 16$}

\addplot[
    mark=square*,
    mark size=1.5pt,
    color=orange!85!black,
    line width=1pt,
    mark options={
        draw=orange!85!black,            
        fill=orange!85!black              
    },
    ] 
    table [x index=0, y index=2] {Data/fraction_vs_L_data_for_different_chi.txt};
    \addlegendentry{$\chi_{\max} = 32$}

    \addplot[
    mark=triangle*, 
    mark size=1.5pt, 
    color=blue!70!black, 
    line width=1pt,   
    ] table [x index=0, y index=3]  {Data/fraction_vs_L_data_for_different_chi.txt};
    \addlegendentry{$\chi_{\max} = 64$}

    \end{axis}



\end{tikzpicture}
    
    \caption{\textbf{Admission rate $\mathcal{R}_{10^{-4}}$ as a function of system size $L$ in the case $\sigma = \infty$.} The admission rate is computed for bond dimensions $\chi_\text{max} = 16, 32, 64$ and over $10,000$ sampled states. 
    }
    \label{fig:fraction_vs_L_for_different_chi}
\end{figure}

\subsection{Different regimes according to $\sigma$ \label{subsec:different_regimes}}
In this section, we investigate the classification of the random states generated by our procedure -- namely, area\hyp law states and volume\hyp law states -- as a function of the standard deviation~$\sigma$.
Here, the distinction between area\hyp law and volume\hyp law states is based on the decay of their eigenvalues: area\hyp law states exhibit exponentially decaying eigenvalues, whereas volume\hyp law states are characterized by a nearly flat eigenvalue distribution.
Firstly, to determine the decay of the eigenvalues for a given $\sigma$, we proceed as follows.
For a set of $n$ eigenvalues, we compute the averaged ordered eigenvalues $\langle \lambda_i \rangle$, and plot them as a function of the index $i$ on a logarithmic scale, as previously shown in~\cref{fig:eigenvalues:uniform-sampling} and the insets of~\cref{fig:coefficient_of_determination}.
We then perform a linear regression on $\log(\langle \lambda_i \rangle)$ as a function of $i$, extracting both the slope $a$ of the fitted line and the corresponding coefficient of determination $R^2$.
For area\hyp law states, the exponential decay of $\langle \lambda_i \rangle$ is expected to result in a linear behavior in $\log(\langle \lambda_i \rangle)$, yielding a negative slope $a$ and a coefficient of determination $R^2$ close to $1$. 
Similarly, for volume\hyp law states, the nearly flat eigenvalue distribution is expected to produce a slope $a$ close to zero, again with a coefficient of determination $R^2$ close to $1$.
Critical states, lying between these two regimes, may exhibit lower values of $R^2$.
By repeating this procedure for different values of $\sigma$, we obtain $a$ and $R^2$ as a function of $\sigma$ for a fixed subsystem dimension $n$.
\Cref{fig:coefficient_of_determination} shows $a$ and $R^2$ as a function of $\sigma$ for $n=2^6$.
The results confirm the expected behavior.
As $\sigma$ goes to zero, the slope $a$ tends toward zero while the coefficient of determination remains close to~$1$, indicating a flat distribution of the eigenvalues characterizing volume\hyp law states. 
For larger values of $\sigma$, the slope $a$ becomes increasingly negative, signalling an exponential decay of the eigenvalues and thus, the onset of the area\hyp law behavior.
Moreover, the coefficient of determination~$R^2$ increases with~$\sigma$, reflecting the improving quality of the exponential fit.
Between those two regimes, the coefficient of determination reaches a minimum at $\sigma_\text{critical}=0.0819$, where $R^2 =0.8355$. 
For a fixed value of $n$, we define $\sigma_\text{critical}$ as
\begin{equation}
    \sigma_\text{critical}=\arg \min_{\sigma}R^2
\end{equation}
and interpret it as the boundary between the volume\hyp law and area\hyp law regimes. 
To visualize how these regimes depend jointly on~$n$ and~$\sigma$, we repeat this analysis for different values of~$n$ and determine the corresponding $\sigma_{\text{critical}}$.
The resulting phase diagram is shown in \cref{fig:Phase_diagram}, where the solid line separates the region generating volume\hyp law states from the region generating area\hyp law states.
\begin{figure*}[htbp]
    \centering
    \begin{tikzpicture}
    \node[anchor=south west,inner sep=0] (image) at (0,0) {
        \begin{tikzpicture}
\pgfplotstableread{
Data/Coeff_determination_more_precised_n64_sigmalimit_0.0001_withtrucationchanging_corrected.txt
}\loadedtable

\begin{axis}[
    width=\textwidth-2cm,
    height=10cm,
    xlabel={Standard deviation $\sigma$},
    ylabel={Slope $a$ of $\bar \lambda_i(\sigma)=a \times i + b$},
    scaled ticks=false,
    yticklabel style={
    /pgf/number format/fixed, 
    /pgf/number format/precision=3, 
    inner sep=3pt,
    },
    xticklabel style={
    /pgf/number format/fixed, 
    /pgf/number format/precision=4, 
    inner sep=3pt,
    },
    grid=both,
    colorbar,
    colorbar style={
    ylabel={Goodness of the fit $R^2$},
    ylabel style={font=\small},
    xlabel style={font=\small},
    yticklabel style={font=\small, inner sep=3pt,},
    xticklabel style={font=\small}},
    colormap/viridis,
    point meta min=0.84,
    point meta max=1.0,
    scatter/use mapped color={
        draw=mapped color,
        fill=mapped color!100,
    },
]

\addplot[
    scatter,
    only marks,
    mark=*,
    mark size=3pt,
    point meta=explicit,
    scatter/use mapped color={
        fill=mapped color,
        draw=black
    },
]
table[
    x index=0,
    y index=1,
    meta index=2,
] {Data/Coeff_determination_more_precised_n64_sigmalimit_0.0001_withtrucationchanging_corrected.txt};
\end{axis}
\end{tikzpicture}
    };

    \begin{scope}[x={(image.south east)}, y={(image.north west)}]

        \node[anchor=south west, inner sep=0] at (0.495,0.55)
        {
        \begin{tikzpicture}
\begin{axis}[
    width=7cm,
    height=4.66cm,
    axis background/.style={fill=white}, 
    xlabel style={opacity=1}, 
    ylabel style={opacity=1}, 
    tick label style={opacity=1}, 
    legend style={draw=none, fill=none}, 
    xlabel={Index $i$},
    ylabel={$\bar\lambda_i (\sigma)$},
    ymode=log,
    xticklabel style={font=\small, inner sep=2pt},
    yticklabel style={font=\small, inner sep=2pt},
    xlabel style={font=\small},
    ylabel style={font=\small},
    legend pos=north east, 
    legend style={font=\small},
]
\addplot[
    mark=*,
    mark size=0.5pt,
    thick,
    color=blue!70!black,
]
table[
    x index=0,
    y index=1,
] {Data/lambda_plot_gaussian_n64_standard_deviation0.3logscale_reglin.txt};

\addlegendentry{$\bar\lambda_i (\sigma=0.3)$}

\addplot[
    red,
    dashed,
    thick,
    domain=1:64,   
    samples=200,
]
{exp(-0.147*x - 2.019)};

\addlegendentry{Linear reg.} 

\end{axis}
\end{tikzpicture}
        };
        
        \node[anchor=south west, inner sep=0] at (0.086,0.08)
        {
        \begin{tikzpicture}
\begin{axis}[
    width=7cm,
    height=4.66cm,
    axis background/.style={fill=white}, 
    xlabel style={opacity=1}, 
    ylabel style={opacity=1}, 
    tick label style={opacity=1}, 
    legend style={draw=none, fill=none}, 
    xlabel={Index $i$},
    ylabel={$\bar\lambda_i (\sigma)$},
    xticklabel style={font=\small, inner sep=2pt},
    yticklabel style={font=\small, inner sep=2pt},
    xlabel style={font=\small},
    ylabel style={font=\small},
    legend pos=north east, 
    legend style={font=\small},
]

\addplot[
    mark=*,
    mark size=0.5pt,
    thick,
    color=blue!70!black,
]
table[
    x index=0,
    y index=1,
] {Data/lambda_plot_gaussian_n64_standard_deviation0.0001logscale_reglin.txt};

\addlegendentry{$\bar\lambda_i (\sigma=0.0001)$}

\addplot[
    red,
    dashed,
    thick,
    domain=1:64,   
    samples=200,
]
{exp(-6*10^(-5)*x - 4.157)}
;

\addlegendentry{Linear reg.} 

\end{axis}
\end{tikzpicture}
        };
        \draw[black, thick] (0.13,0.935) rectangle (0.17,0.875);
        \draw[black, thick] (0.8,0.18) rectangle (0.84,0.12);

        \draw[->, thick] (0.15,0.875) -- (0.27,0.48); 
        \draw[->, thick] (0.82,0.18) -- (0.7,0.53);

    \end{scope}
\end{tikzpicture}
    \caption{\textbf{Slope $a$ and coefficient of determination $R^2$ of the linear regression $\log(\langle \lambda_i\rangle )=a\times i+b$ as a function of the standard deviation $\sigma$.} 
    Here, $n=2^{6}$. 
    The worst fit is obtained at $\sigma_\text{critical}=0.0819$, where $R^2=0.8355$. 
    The insets show $\langle \lambda_i \rangle$ as a function of the index $i$ for $\sigma=0.0001$ (bottom left) and $\sigma=0.3$ (top right). 
    For $\sigma=0.0001$, the eigenvalue distribution is nearly flat, this is characteristic of volume\hyp law states. The linear regression is $y=-6 .10^{-5}\times i-4.157$.  
    While for $\sigma=0.3$, the eigenvalues exhibit an exponential decay, characteristic of area\hyp law states. The linear regression is $y=-0.147 \times i-2.019$.}
    \label{fig:coefficient_of_determination}
\end{figure*}

%

\begin{figure}[]
    \centering
    \begin{tikzpicture}
    \begin{axis}[
        xlabel={$n$},
        ylabel={$\sigma_\text{critical}$},
        xlabel style={},
        ylabel style={},
        xticklabel style={},
        yticklabel style={},
        legend pos=north east,
        legend style={},
        every axis plot/.append style={line width=1pt},
        clip=false, 
        scaled ticks=false, 
        yticklabel style={
        /pgf/number format/fixed, 
        /pgf/number format/precision=4, 
        },
    ]

    \addplot[
    mark=*, 
    mark size=1pt, 
    color=blue!70!black, 
   ] table [x index=0, y index=1] {Data/critical_standard_deviation.txt};

    \node[anchor=south west] at (rel axis cs:0.3, 0.2) {Area law};

    \node[anchor=south west] at (rel axis cs:0.01, 0.03) {Volume law};

    \end{axis}
    \begin{axis}[
        at={(main plot.north east)}, 
        anchor=north east, 
        xshift=-0.3cm, yshift=-0.3cm, 
        width=6cm, 
        height=4cm, 
        xmin=0, xmax=0.1, 
        ymin=0, ymax=1, 
        xlabel={$\sigma$}, 
        ylabel={$\mathcal{R}_{10^{-3}}$}, 
        grid=both, 
        scaled ticks=false,
        xticklabel style={
                /pgf/number format/fixed,
                /pgf/number format/precision=4,
            },
        tick label style={font=\tiny}, 
        xlabel style={font=\tiny}, 
        ylabel style={font=\tiny}, 
        axis background/.style={fill=white}, 
    ]

    \addplot[
        mark=triangle*, 
        mark size=1.5pt, 
        color=blue!70!black, 
        line width=1pt,
    ] table [x index=0, y index=1] {Data/fraction_vs_sigma.txt};

    \end{axis}

\end{tikzpicture}
    \caption{\textbf{Phase diagram of the different regimes as a function of the standard deviation $\sigma$ and the subsystem dimension $\mathbf{n}$.} 
    The solid line represents $\sigma_\text{critical}$, separating the regime generating area\hyp law states (above the line) from the regime generating volume\hyp law states (below the line).  
    \textbf{(Inset)} Admission rate $\mathcal{R}_{10^{-3}}$ as a function of the standard deviation $\sigma$ for a fixed system size $L=12$ and bond dimension $\chi_\text{max} = 64$, computed over $10,000$ states.}
    \label{fig:Phase_diagram}
\end{figure}
Secondly, we compute the success rate of our procedure in constructing a compatible global pure state from the set of eigenvalues as a function of $\sigma$ in order to understand how its success rate evolves across these different regimes. 
To this end, we study the admission rate defined in~\cref{eq:admission_rate} as a function of $\sigma$ for a system of $L$ lattice sites and put it into perspective with the phase diagram.
Although $\sigma_{\text{critical}}$ depends on the subsystem dimension~$n$, this dependence weakens for sufficiently large values of~$\sigma$ and subsystem sizes larger than six (see \cref{fig:bond_dimension_3D_surface_plot_average_over_10_samples_until_9_gaussian_disk_log}).
For simplicity, we therefore use a fixed value of~$\sigma$ for all sets of eigenvalues associated with a given state.
The inset of \cref{fig:Phase_diagram} shows the admission rate $\mathcal{R}_{10^{-3}}$ as a function of $\sigma$ for $L=12$.
The results reveal that the algorithm more frequently finds compatible global states for small values of $\sigma$, corresponding to the volume\hyp law regime. 
This behavior is expected, since area\hyp law states occupy a Haar measure zero subset of the Hilbert space~\cite{Page_1993}, making them intrinsically harder to generate.
Thus, these findings are consistent with the phase diagram. 
For small~$\sigma$, the eigenvalue distributions are nearly flat and the algorithm performs efficiently, whereas for large~$\sigma$, the eigenvalues decay exponentially and the construction of compatible global states becomes computationally more demanding.
Nevertheless, our algorithm remains capable of generating such area\hyp law states.
\section{Discussion}
In this article, we introduced ensembles of random pure quantum states -- the $\sigma$\hyp ensembles -- and established some of their key properties. 
While the Haar measure already provides a well\hyp known method for generating random states, Haar-random states obey a volume law for entanglement entropy \cite{Page_1993}. 
This makes them unsuitable for classical simulation and unrepresentative of the quantum states typically encountered in quantum many\hyp body systems, which generally exhibit area\hyp law entanglement scaling.
Here, we proposed constraining the entanglement of the subsystems by enforcing a specific probability distribution on their eigenvalues.
We then presented a method for constructing a compatible global pure state using the~\gls{MPS} formalism. 
Our results demonstrate that the class of random states can be tuned between area\hyp law and volume\hyp law behaviors through a judicious choice of the width of the eigenvalue distribution around the maximally entangled point.
At first consideration, the exponential decay of the admission rate with $L$ must have appeared problematic.
However, an analysis of the admission rate as a function of the width of the eigenvalue distribution reveals that the prefactor of this exponential decay can be significantly reduced.
Consequently, there exists a regime in which random area\hyp law states can be efficiently sampled.
Our findings extend those of Collins et al. which show that, for random graph states, the area law holds generally for some special cases and asymptotically for the general case \cite{Collins_2010, Collins_2013}.
However, unlike Haar-random states and our $\sigma$\hyp ensembles, the topology of the couplings between the physical subsystem of the random graph states is dictated by the underlying graph topology.
Most notably, this work presents, to our knowledge, the first method capable of generating random quantum states obeying the area law, beyond the specific case of random graph states.
The key achievement of our method is its ability to produce random quantum states with a controlled entanglement growth throughout the bipartitions.
In the context of classical simulation of quantum systems, this avoids the intractability issues associated with Haar-random pure states.
Furthermore, our random quantum states may also be more relevant than volume\hyp law random quantum states for testing, benchmarking, validating, and evaluating the performance of quantum algorithms, given the physical relevance of area-law states \cite{Hastings_2007, Wolf_2008, Aharonov_2011, Kull_2019, Shor_1997, Grover_1996, Boixo2014}.
Our area\hyp law random quantum states lack the unitary invariance that characterizes the Haar measure.
However, we believe that imposing a different probability distribution on the eigenvalues of successive subsystems could yield a unitary\hyp invariant distribution of area\hyp law states. 
Such a distribution would extend the Haar measure to the space of area\hyp law states, which has zero measure under the Haar distribution.
We also believe that our method could enable the generation of critical states. 
Indeed, while we provide a phase diagram of the area\hyp law and volume\hyp law regimes as a function of system dimensionality $n$ and the width of the Gaussian distribution around the maximally entangled point on the $n$\hyp sphere, critical states are expected to exist at the boundary between these regimes. 
Future work should explore alternative eigenvalue distributions and map the full class of random states as a function of their distribution on the $n$\hyp sphere.
Finally, establishing additional physical properties of the $\sigma$\hyp ensembles remains a subject of further investigation.

\section*{Acknowledgments}
The authors acknowledge the Munich Quantum Valley, funded by the Free State of Bavaria. SP also acknowledges support by the Deutsche Forschungsgemeinschaft (DFG, German Research Foundation) under Germany’s Excellence Strategy-426 EXC-2111-39081486.
\newpage
\bibliographystyle{unsrt}
\bibliography{bibliography}

\clearpage
\onecolumngrid
\appendix
\section{Generating a set of eigenvalues for a uniform sampling on the unit sphere \label{annex:generating_the_eigenvalues_uniform}}

To generate random eigenvalues, we sample a random point on the positive orthant of the unit $n$-sphere and use the square of its Cartesian coordinates as the eigenvalues. 
This ensures that the eigenvalues are non\hyp negative and sum to one. 
Restricting the sampling to the positive orthant does not alter the distribution of the eigenvalues, since the eigenvalues are the square of the coordinates.
This restriction simply eliminates duplicate sets of eigenvalues.
In the appendices, we use bold capitalized letters to denote random variables and lowercase, non\hyp bold letters to refer to their values.
To generate a random point uniformly on the positive orthant of the unit $n$-sphere, we sample $n-1$ random variables $\bm{\varPhi_1}, \ldots , \bm{\varPhi_{n-1}}$, where $\bm{\varPhi_1}, \ldots , \bm{\varPhi_{n-2}}$ are uniformly distributed on $[0,\pi/2]$ and $\bm{\varPhi_{n-1}}$ is uniformly distributed on $[0,\pi/2[$. 
These variables represent the spherical coordinates of a point on the $n$-sphere. 
We convert these spherical coordinates to Cartesian coordinates as follows:
\begin{align*}
    \bm{X_1} &= \cos(\bm{\varPhi_1}), \label{equ:coordinate_1}  \\
    \bm{X_{2}}&=\sin(\bm{\varPhi_{1}})\cos(\bm{\varPhi_{2}}),\\
    \bm{X_{3}}&=\sin(\bm{\varPhi_{1}})\sin(\bm{\varPhi_{2}})\cos(\bm{\varPhi_{3}}),\\
    & \qquad \vdots \\
    \bm{X_{n-1}}&=\sin(\bm{\varPhi_{1}})\ldots \sin(\bm{\varPhi_{n-2}})\cos(\bm{\varPhi_{n-1}}),\\
    \bm{X_{n}}&=\sin(\bm{\varPhi_{1}})\ldots \sin(\bm{\varPhi_{n-2}})\sin(\bm{\varPhi_{n-1}}).
\end{align*}
Finally, the eigenvalues are obtained by squaring these Cartesian coordinates:
\begin{equation*}
    \bm{\Lambda_i} = \bm{X_i}^2 \qquad \forall~ 1 \leq i \leq n \;.
\end{equation*}
\section{Expectation values of the eigenvalues for a uniform sampling on the unit sphere \label{annex:computation_of_exponential_decay}}
In this appendix, we demonstrate the exponential decay of the expectation values of the non-ordered eigenvalues $\bm{\Lambda_1}, \ldots, \bm{\Lambda_n}$. 
For $i \in \llbracket 1,n-1  \rrbracket$, the expectation value of $\bm{\Lambda_1}$ is:
\begin{equation*}
      \mathbb{E}[\bm{\Lambda_i}] = \mathbb{E}[\bm{X_i}^2] 
\end{equation*}
Using the relationship between Cartesian and spherical coordinates, we have:
\begin{equation*}
  \mathbb{E}[\bm{\Lambda_i}] = \mathbb{E}[\cos(\bm{\varPhi_i})^2\prod_{j=1}^{i-1} \sin(\bm{\varPhi_j})^2] 
\end{equation*}
By the independence of the random variables $\bm{\varPhi_1}, \ldots , \bm{\varPhi_{n-1}}$, this simplifies to:
\begin{equation*}
 \mathbb{E}[\bm{\Lambda_i}] = \mathbb{E}[\cos(\bm{\varPhi_i})^2]\prod_{j=1}^{i-1} \mathbb{E}[\sin(\bm{\varPhi_j})^2] .
\end{equation*}
For $i=n$, the expectation value is:
\begin{equation*}
    \mathbb{E}[\bm{\Lambda_n}] = \prod_{j=1}^{n-1} \mathbb{E}[\sin(\bm{\varPhi_j})^2] \;.
\end{equation*}
For $1 \leq j \leq n-1$, the~\gls{PDF} of $\bm{\varPhi_j}$ is 
\begin{equation*}
p(\bm{\varPhi_j} = \varphi)=
    \begin{cases} \frac{2}{\pi}  & \text{if } \varphi \in [0,\pi/2], \\
    0 & \text{otherwise.} 
    \end{cases} 
\end{equation*}
Using the identity \(\cos^2(x) = \frac{1 + \cos(2x)}{2}\), we compute:

\begin{align*}
\mathbb{E}[\cos^2(\bm{\varPhi})] &= \int_0^{\frac{\pi}{2}} \cos^2(\varphi) \cdot \frac{2}{\pi} \, d\varphi \\
&= \frac{2}{\pi} \int_0^{\frac{\pi}{2}} \frac{1 + \cos(2\varphi)}{2} \, d\varphi \\
&= \frac{1}{2}\;.
\end{align*}
\noindent From this, we compute:
\begin{equation*}
    \mathbb{E}[\sin^2(\bm{\varPhi})] = 1 - \mathbb{E}[\cos^2(\bm{\varPhi})] = \frac{1}{2}\;.
\end{equation*}
For $i \leq n-1$, the expectation value of the $i$-th eigenvalue is
\begin{equation*}
    \mathbb{E}[\bm{\Lambda_i}] = \left( \frac{1}{2} \right) \cdot \left( \frac{1}{2} \right)^{i-1} = \left( \frac{1}{2} \right)^i \; ,
\end{equation*}
and for $i = n$,
\begin{equation*}
    \mathbb{E}[\bm{\Lambda_n}] = \left( \frac{1}{2} \right)^{n-1} \;.
\end{equation*}
This demonstrates that the eigenvalues, obtained by sampling the coordinates uniformly on the positive orthant of the unit $n$-sphere and squaring them, decay exponentially with $i$. 
Thus, the generated quantum states obey the area law.
\section{Expectation value of the von Neumann entropy of the states for a uniform sampling on the unit sphere \label{appendix:ananlytical_proof_average_von_neumann_entropy}}

In this appendix, we compute the expectation value of the von Neumann entropy of the states produced by the distribution defined in~\cref{sec:sampling:general} in the case where the points generating the eigenvalues are uniformly sampled from the positive orthant of the unit $n$-sphere. To do so, we first derive the~\gls{PDF} of the non\hyp ordered eigenvalues obtained by uniformly sampling points on the positive orthant of the unit $n$-sphere, and then calculate the von Neumann entropy.

\subsection{\Gls{PDF} of the eigenvalues
\label{annex:pdf_eigenvalues}}
The Cartesian coordinates $(x_1,x_2,\ldots, x_n)$ of a point on the unit $n$-sphere are related to its spherical coordinates as follows~\cite{n_dim_sphere}:
\begin{align*}
    x_1 &= \cos(\varphi_1),  \\
    x_{2}&=\sin(\varphi _{1})\cos(\varphi_{2}),\\
    x_{3}&=\sin(\varphi _{1})\sin(\varphi _{2})\cos(\varphi _{3}),\\
    & \qquad \vdots \\
    x_{n-1}&=\sin(\varphi _{1})\ldots \sin(\varphi _{n-2})\cos(\varphi _{n-1}),\\
    x_{n}&=\sin(\varphi_{1})\ldots \sin(\varphi_{n-2})\sin(\varphi_{n-1}).
\end{align*}
where $\varphi_1,\varphi_2, \ldots, \varphi_{n-2}  \in [0,\pi ]$ and $\varphi_{n-1} \in [0,2\pi[$. 
The inverse transformation is given by:
\begin{equation}
\begin{split}
    \varphi_1 &= \operatorname{atan2}(\sqrt{x_n^2+ x^2_{n-1} + \ldots + x_2^2}, x_1)  \\
    \varphi_{2}&= \operatorname{atan2}(\sqrt{x_n^2 + x^2_{n-1}+ \ldots + x_3^2}, x_2)\\
    & \qquad \vdots \\
    \varphi_{n-2}&= \operatorname{atan2}(\sqrt{x_n^2 + x_{n-1}^2}, x_{n-2})\\
    \varphi_{n-1}&=\operatorname{atan2}(x_n , x_{n-1})=\operatorname{atan}(1-\sum_{i=1}^{n-1}x_i, x_{n-1}).
\label{eq:cartesian_to_spherical_coordinates}
\end{split}
\end{equation}
Here, $\operatorname{atan2}$ denotes the two-argument arctangent function, defined by:
\begin{equation*}
    \operatorname{atan2}(y, x) =
\begin{cases} 
    \arctan(\frac{y}{x}), & \text{if } x > 0, \\
    \arctan(\frac{y}{x})+\pi, & \text{if } x < 0 \text{ and } y \geq 0, \\
    \arctan(\frac{y}{x})-\pi, & \text{if } x < 0 \text{ and } y < 0, \\
    +\frac{\pi}{2}, & \text{if } x = 0 \text{ and } y > 0, \\
    -\frac{\pi}{2}, & \text{if } x = 0 \text{ and } y < 0, \\
    \text{undefined} & \text{if } x = 0 \text{ and } y = 0.
\end{cases}
\end{equation*}
This inverse transform is not unique in some special cases. 
For instance, if $x_k = \ldots = x_n  =0 $, then the angles $\varphi_{k+1}, \ldots, \varphi_{n-1} $ are undefined. 
In such cases, these angles may be set to zero to resolve the ambiguity. 
To derive the\gls{PDF} of the eigenvalues, we proceed via two successive transformations.
First, we express the transformation from spherical coordinates to Cartesian coordinates $\vec{G}$:
\begin{align*}
    \centering
    \begin{split}
    G_1(\bm{\varPhi_1}, \ldots, \bm{\varPhi_{n-1}})  = & \cos \bm{\varPhi_1} = \bm{X_1} \\
    G_2(\bm{\varPhi_1}, \ldots, \bm{\varPhi_{n-1}}) = & \sin \bm{\varPhi_1} \cos \bm{\varPhi_2} = \bm{X_2} \\
    & \vdots \\
    G_{n-1}(\bm{\varPhi_1}, \ldots, \bm{\varPhi_{n-1}}) = & \sin \bm{\varPhi_1} \sin \bm{\varPhi_2} \ldots \sin \bm{\varPhi_{n-2}} \cos \bm{\varPhi_{n-1}} =  \bm{X_{n-1}} \;.
    \end{split}
    \label{eq:expression_of_the_transformation_G}
\end{align*}
The remaining coordinate $\bm{X_n}$ is uniquely determined (up to a sign) by the constraint $\sum_{i=1}^n x_i^2 = 1$.
Restricting to the positive orthant fixes this sign, so $\bm{X_n}$ is completely determined by $x_1, \ldots, x_{n-1}$.
Accordingly, we exclude $\bm{X_n}$ from the transformation to ensure equal dimensions between the Cartesian and spherical coordinate systems.
We adopt the shorthand notation $p(x) \equiv p(\bm{X} = x)$ when no ambiguity arises.
According to \cite{devore2007modern}, the~\gls{PDF} of $\vec{G}$ can be evaluated if $\vec{G}$ is a bijective differentiable function:
\begin{equation}
    p(x_1, \ldots, x_{n-1}) = p(\vec{G}^{-1}(x_1, \ldots, x_{n-1})) \Bigg|\det{\frac{d\vec{G}^{-1}}{d\vec{z}}\Bigg|_{\vec{z}=(x_1, \ldots, x_{n-1}) }}\Bigg|\;,
\label{eq:vector_to_vector_change_variable_density_fct}
\end{equation}
where $p(\vec{G}^{-1}(x_1, \ldots, x_{n-1}))$ is the joint probability of the random vector $(\varphi_1, \ldots , \varphi_{n-1})$ and the differential is the Jacobian of the inverse of $\vec{G}$, evaluated at $(x_1, \ldots, x_{n-1})$. 
The Jacobian matrix of $\Vec{G}$ is given by:
\begin{align}
J_n &  \equiv \frac{d\vec{G}^{-1}}{d\vec{z}}\Bigg|_{\vec{z}=(x_1, \ldots, x_{n-1})} \nonumber \\
& = \begin{bmatrix}
-\sin(\varphi_1) & 0 & \ldots & 0 \\
\cos(\varphi_1)\cos(\varphi_2) & -\sin(\varphi_1)\sin(\varphi_2) & \ddots & \vdots \\
\vdots  &  & \ddots & 0  \\
\cos(\varphi_1) \sin(\varphi_2)\ldots \sin(\varphi_{n-2})\cos(\varphi_{n-1})& \ldots & & -\sin(\varphi_1)\ldots \sin(\varphi_{n-1})
\end{bmatrix}. \label{eq:jacobian_matrix_of_G}
\end{align}
When restricted to the domain
\begin{equation}
\vec{G} :\; ]0,\frac{\pi}{2}[^{n-1} \longrightarrow \mathcal{S}^n_+,
\label{eq:G_bijective}
\end{equation}
where
\begin{equation}
\mathcal{S}^n_+ =
\{ x \in \mathbb{R}^n \mid \|x\|_2 = 1,\; x_i > 0 \; \forall i \in \llbracket 0, n+1 \rrbracket \},
\label{eq:definition_of_S_tilde_n}
\end{equation}
the mapping $\vec{G}$ is bijective and continuously differentiable.
\begin{lemma}
For the $n$-sphere: 
\begin{equation*}
    \left| \det{\frac{\partial \vec{G}^{-1}(\vec{z}) }{\partial \vec{z}} \Big|_{\vec{z}=(x_1, \ldots,x_{n-1})} } \right| = \prod_{k=1}^{n-1} \frac{1}{\sqrt{1-\sum_{i=1}^k x_i^2}} \;.
\end{equation*}
\end{lemma}

\begin{proof} The inverse transformation of $\vec{G}$ is given by:
\begin{align*}
    \vec{G}^{-1}:\tilde{S}^n & \longrightarrow ]0,\frac{\pi}{2}[^{n-1}\\
    (x_1, \ldots, x_{n-1}) & \mapsto \begin{pmatrix}
\varphi_1 = \arctan \left( \frac{\sqrt{1 - x_1^2}}{x_1} \right) \\
\varphi_2 = \arctan \left( \frac{\sqrt{1 - x_1^2 - x_2^2}}{x_2} \right) \\
\vdots \\
\varphi_{n-1} = \arctan \left( \frac{\sqrt{1 - \sum_{i=1}^{n-1}x_i^2}}{x_{n-1}} \right)
\end{pmatrix}^T
\end{align*}
For $n = 2$,
\begin{align*}
    \vec{G}^{-1}:\tilde{S}^n & \longrightarrow ]0,\frac{\pi}{2}[ \\
    x_1 & \mapsto \varphi_1 = \arctan \left( \frac{\sqrt{1 - x_1^2}}{x_1} \right) 
\end{align*}
Using $\arctan'(x)=\frac{1}{(1+x^2)}$, we get: 
\begin{align*}
    \left( \frac{\partial \vec{G}^{-1}(z) }{\partial z} \Big|_{z=x_1 } \right)  & = \frac{-1}{\sqrt{1-x_1^2}} \;.
\end{align*}
Thus,
\begin{equation*}
    \left| \det{\frac{\partial \vec{G}^{-1}(z) }{\partial z} \Big|_{z=x_1} } \right|=\frac{1}{\sqrt{1-x_1^2}}\;.
\end{equation*}
We now prove the lemma by induction. 
We assume it holds for the $n$-sphere, and show that it holds for the $(n+1)$\hyp sphere. 
$\left( \frac{\partial \vec{G}^{-1}(\vec{z}) }{\partial \vec{z}} \Big|_{z=x_1 } \right)$ is a lower triangular matrix.
Thus, the determinant is the product of the diagonal elements. 
Let $J_n^{-1}$ be the Jacobian of the inverse transformation of $\vec{G}^{-1}$ associated with the $n$-sphere coordinates, that is, mapping $(x_1, \ldots, x_{n-1})$ to $(\varphi_1, \ldots, \varphi_{n-1})$. It is a matrix of dimension $(n-1) \times (n-1)$. 
We have: 
\begin{equation*}
J^{-1}_{n+1} = 
\begin{bmatrix}
J^{-1}_n & \mathbf{0} \\
\mathbf{v}^\top & \alpha
\end{bmatrix}\; ,
\end{equation*}
where $\mathbf{0}$ is a column vector of zeros with  $n-1$ entries, $\mathbf{v}^\top $ is a row vector of non-zero elements of length $n-1$ and $\alpha = \frac{\partial \vec{G}^{-1}_{n}}{\partial x_{n}} =  \frac{-1}{\sqrt{1-\sum_{i=1}^{n}x_i^2}} $ is the diagonal element at position $(n, n)$.
Using the induction hypothesis, we get: 
\begin{equation*}
    |\det{J_{n+1}^{-1}}| = |\det{J_{n}^{-1}} \cdot \alpha| =  \prod_{k=1}^{n-1} \frac{1}{\sqrt{1-\sum_{i=1}^k x_i^2}} \cdot \frac{1}{\sqrt{1-\sum_{i=1}^{n}x_i^2}}  \;,
\end{equation*}
which completes the induction. 
\end{proof}
Moreover, since the spherical coordinates are independent and uniformly distributed on the interval $]0,\frac{\pi}{2}[$,  their joint probability distribution is given by: 
\begin{equation*}
    p(\varphi_1,\ldots, \varphi_{n-1}) =
\begin{cases} 
    \big(\frac{2}{\pi}\big)^{n-1} & \text{if } ~\forall 1 \leq i \leq n -1 ,~ \varphi_i \in ]0,\frac{\pi}{2}[, \\
    0 & \text{otherwise.}  
\end{cases}
\end{equation*}
We can now calculate the~\gls{PDF} of the Cartesian coordinates using~\cref{eq:vector_to_vector_change_variable_density_fct}. 
We substitute the expression for the joint probability distribution of the spherical coordinates,
\begin{equation*}
    p(\varphi_1,\ldots, \varphi_{n-1}) =p(\vec{G}^{-1}(x_1, \ldots, x_{n-1})) \;,   
\end{equation*}
and the determinant of the inverse of the Jacobian, 
\begin{equation*}
\det{\frac{d\vec{G}^{-1}(\vec z)}{d\vec{z}}\Bigg|_{\vec{z}=(x_1, \ldots, x_{n-1}) }}\;,  
\end{equation*} 
to find the~\gls{PDF} of the Cartesian coordinates:
\begin{equation*}
    p(x_1, \ldots, x_{n-1})=
    \begin{cases} 
    \big(\frac{2}{\pi}\big)^{n-1} \prod_{k=1}^{n-1} \frac{1}{\sqrt{1-\sum_{i=1}^k x_i^2}}& \text{ if } ~\forall 1 \leq i \leq n -1 ,~ x_i \in ]0,1[  \text{ and } \sum_{i=1}^{n-1} x_i^2 < 1,  \\
    0 & \text{ otherwise.}  
    \label{eq:probability_distribution_our_eigenvalues}
\end{cases}
\end{equation*}
Finally, the eigenvalues are obtained via the mapping:
\begin{align*}
    \vec{G}_\text{square}:\tilde{S}^n & \longrightarrow \text{Int}(\Delta^{n-1})\\
    (x_1, \ldots, x_{n-1}) &  \mapsto (x_1^2,\ldots,x_{n-1}^2)\;.
\end{align*}
Here, $\text{Int}(\Delta^{n-1})$ is the interior of $\Delta^{n-1}$, the standard $(n-1)$\hyp dimensional simplex:
\begin{equation*}
    \Delta^{n-1} = \{ x \in \mathbb{R}^{n} |  \sum_{i=1}^{n} x_i = 1, \quad x_i \geq 0 \quad \text{ for } i = 1, \ldots, n \}.
\end{equation*}
The inverse transformation of $\vec{G}_\text{square}$ is:
\begin{align*}
    \vec{G}_\text{square}^{-1}: \text{Int}(\Delta^{n-1}) & \longrightarrow \tilde{S}^n \\
    (\lambda_1, \ldots, \lambda_{n-1}) & \mapsto (\sqrt{\lambda_1}, \ldots, \sqrt{\lambda_{n-1}} )\;.
\end{align*}
The Jacobian of the inverse transformation is given by:
\begin{equation*}
    \frac{\partial \vec{G}_\text{square}^{-1}(\vec z)}{\partial \vec{z}}\Big|_{\vec{z}=(\lambda_1, \ldots, \lambda_{n-1})} = \operatorname{diag}\left(\frac{1}{2\sqrt{\lambda_1}}, \ldots, \frac{1}{2\sqrt{\lambda_{n-1}}} \right).
\end{equation*}
Using the same formula as in~\cref{eq:vector_to_vector_change_variable_density_fct}, we obtain the following result.
\begin{proposition}
    The joint \gls{PDF} of the non\hyp ordered eigenvalues $(\lambda_1, \ldots, \lambda_{n-1})$ is given by
\begin{equation}
    p(\lambda_1, \ldots, \lambda_{n-1})=
    \begin{cases} 
    \big(\frac{1}{\pi}\big)^{n-1} \prod_{k=1}^{n-1} \frac{1}{\sqrt{\lambda_k\big(1-\sum_{i=1}^k \lambda_i\big)}}& \text{ if } ~\forall 1 \leq i \leq n -1 ,~ \lambda_i \in ]0,1[  \text{ and } \sum_{i=1}^{n-1} \lambda_i < 1,  \\
    0 & \text{ otherwise.}  
\end{cases}
\label{eq:joint_probability_distribution_eigenvalues}
\end{equation}
\end{proposition}
\subsection{Expectation value of the von Neumann entropy \label{appendix:average_von_neumann_entropy}}
For a quantum state with density matrix $\rho$, the von Neumann entropy is given by~\cite[p.273]{Bengtsson_Zyczkowski_2006}:
\begin{equation*}
    S=-\operatorname{Tr} (\rho \ln \rho )\;,
\end{equation*}
where $\ln$ denotes the natural logarithm of a matrix. 
Expressing $\rho$ in its eigenbasis, $\rho = \sum_{j} \lambda_{j} \ket{j}\bra{j}$, the von Neumann entropy reduces to
\begin{equation*}
    S = -\sum_{j} \lambda_{j} \ln \lambda_{j}.
\end{equation*}
For a random quantum state of dimension $n$ drawn from our distribution, the expectation value of the von Neumann entropy reads:
\begin{align*}
   \langle \bm{S_n} \rangle & =  \langle -\sum_{j=1}^n \bm{\Lambda_{j}}\ln \bm{\Lambda_{j}} \rangle \\
   & =   -\sum_{j=1}^n  F_j \;,
\end{align*}
where $\langle \cdot \rangle$ denotes the expectation value with respect to our distribution, and $F_j \equiv \langle \bm{\Lambda_{j}}\ln \bm{\Lambda_{j}} \rangle$. 
First, we calculate $F_1$. 
To do so, we use the expression of $p(\lambda_1)$ derived in~\cref{annex:pdf_eigenvalues} and given by~\cref{eq:joint_probability_distribution_eigenvalues}.
Substituting this expression into the definition of $F_1$, we obtain:
\begin{equation*}
    F_1 = \int_0^1 p(\lambda_1)\lambda_1 \ln\left(\lambda_1\right)   d \lambda_1 = \int_0^1 \frac{\lambda_1 \ln\left(\lambda_1\right) }{\pi \sqrt{\lambda_1(1 - \lambda_1)}}  d \lambda_1.
\end{equation*}
Using the substitution $\sin^2 \theta = \lambda_1$, we have 
$d \lambda_1 = 2 \sin \theta \cos \theta d \theta$,
and the integration limits transform as:
\begin{equation*}
\lambda_1 = 0 \implies \theta = 0, \qquad \lambda_1 = 1 \implies \theta = \frac{\pi}{2}.
\end{equation*}
Thus, the integral becomes
\begin{align*}
    F_1 & =\frac{1}{\pi} \int_0^{\pi/2} \frac{\sin^2 \theta \ln \left( \sin^2 \theta \right)}{\sin \theta \cos \theta} \cdot 2 \sin \theta \cos \theta d \theta\\
    & =\frac{4}{\pi} \int_0^{\pi/2} \sin^2 \theta \ln \left( \sin \theta \right) d \theta .
\end{align*}
Using the trigonometric identity $\sin^2 \theta = \frac{1 - \cos (2\theta)}{2}$, we can rewrite the integral as:
\begin{align*}
F_1 = \frac{2}{\pi} \left( \int_0^{\pi/2} \ln (\sin \theta) d \theta - \int_0^{\pi/2} \cos (2 \theta) \ln (\sin \theta) d \theta \right).
\end{align*}
From~\cite[II~629,643]{fikhtengolts1947}, we have:
\begin{equation*}
    \int_0^{\pi/2} \ln (\sin \theta) d \theta = - \frac{\pi}{2}
\ln 2 \;.
\end{equation*}
Moreover, from~\cite[p.582,~4.384 Eq.~7]{gradshteyn2007}, we have:
\begin{equation*}
    \int_0^{\pi/2} \cos (2 n \theta) \ln (\sin \theta) d \theta = \left\{
    \begin{array}{ll}
        -\frac{\pi}{4n}, & \mbox{for } n >0 \\
        -\frac{\pi}{2} \ln 2, & \mbox{for } n=0 \; .
    \end{array}
\right. 
\end{equation*}
In particular, for the case $n=1$,
\begin{equation*}
    \int_0^{\pi/2} \cos (2 \theta) \ln (\sin \theta) d \theta =-\frac{\pi}{4} \;.
\end{equation*}
Combining these results, we obtain: 
\begin{equation*}
    F_1 = \frac{2}{\pi}(-\frac{\pi}{2}\ln(2)+\frac{\pi}{4}) = \frac{1}{2}-\ln2 = -a \;,
\end{equation*}
where $a=\ln(2)-\frac{1}{2} > 0$.
We now calculate $F_2$:
\begin{align*}
F_2  & = \int_0^1 p(\lambda_2)\lambda_2 \ln\left(\lambda_2\right)   d \lambda_2 \\& 
= \frac{1}{\pi^2}\int_0^1 \lambda_2 \ln\left(\lambda_2\right)   d \lambda_2 \int_0^{1-\lambda_2}  \frac{1 }{ \sqrt{\lambda_1(1 - \lambda_1)} \sqrt{\lambda_2(1-\lambda_1-\lambda_2)}}  d \lambda_1 \\& 
= \frac{1}{\pi^2}\int_0^1 \frac{1}{\sqrt{\lambda_1(1 - \lambda_1)}}   d \lambda_1 \int_0^{1-\lambda_1}  \frac{\lambda_2 \ln\left(\lambda_2\right) }{  \sqrt{\lambda_2(1-\lambda_1-\lambda_2)}}  d \lambda_1 \;.
\end{align*}
Let us define the following integral, where $s=\lambda_1$:
\begin{align}
A(s) & \equiv \int_0^{1-s} \frac{\lambda_2 \ln \lambda_2}{\sqrt{\lambda_2(1 -s- \lambda_2)}} d \lambda_2 \label{eq:A_recursion_F_2} \\
&= \int_0^{1-s} \frac{\lambda_2 \ln \lambda_2}{\sqrt{\frac{\lambda_2}{{1-s}}(\frac{1 -s}{1-s}- \frac{\lambda_2}{1-s})}}  \frac{d \lambda_2}{1-s} \nonumber \;.
\end{align}
Using the substitution: $x=\frac{\lambda_2}{1 - s}$, we have $dx=\frac{d\lambda_2}{1 - s}$ and the integration limits transform as:
\[
\lambda_2 = 0 \implies x = 0, \qquad \lambda_2 = 1-s \implies x = 1.
\]
Thus, the integral becomes
\begin{align*}
    A(s)  & = (1-s) \int_0^1 \frac{x \ln (x(1 - s))}{\sqrt{x(1-x)}} dx \\
    & = (1-s) \Big( \int_0^1 \frac{x \ln (x)}{\sqrt{x(1-x)}} dx + \ln(1 - s)  \int_0^1 \sqrt{ \frac{x }{1-x}} dx \Big) \;.
\end{align*}
The first term in the parentheses is $-\pi a$ from the calculation of $F_1$, and the second term is $\ln(1-s)\frac{\pi}{2}$.
Therefore, 
\begin{equation*}
A(s)  = \pi(1-s)\left(-a + \frac{1}{2}\ln(1-s)\right).\end{equation*}
Substituting this expression back into $F_2$, we obtain:
\begin{align*}
    F_2 & = \frac{1}{\pi} \int_0^1 \frac{(1-\lambda_1)}{\sqrt{\lambda_1 (1-\lambda_1)}}  \Big( -a  + \ln(1 - \lambda_1) \frac{1}{2 }\Big) d\lambda_1\\
    & =  \frac{1}{\pi} \Big( -a \int_0^1 \sqrt{ \frac{1-\lambda_1}{\lambda_1 }} d\lambda_1   + \frac{1}{2}\int_0^1 \frac{(1-\lambda_1)}{\sqrt{\lambda_1 (1-\lambda_1)}}  \ln(1 - \lambda_1) \Big)d\lambda_1 \\
    & =  \frac{1}{\pi} \Big( - \frac{a \pi}{2} + \frac{1}{2}\int_0^1 \frac{(1-\lambda_1)}{\sqrt{\lambda_1 (1-\lambda_1)}}  \ln(1 - \lambda_1) \Big)d\lambda_1 \;.
\end{align*}
Using the substitution $\lambda=1-\lambda_1$ in the remaining integral, we have $d\lambda=-d\lambda_1$, and the integration limits transform as:
\begin{equation*}
    \lambda_1 = 0 \implies \lambda = 1, \qquad \lambda_1 = 1 \implies \lambda = 0.
\end{equation*}
Thus, the integral becomes
\begin{equation*}
    F_2 =  \frac{1}{\pi} \Big( - \frac{a \pi}{2} + \frac{1}{2}\int_0^1 \frac{\lambda}{\sqrt{\lambda (1-\lambda)}}  \ln( \lambda) \Big)d\lambda \; .
\end{equation*}
Recognizing the same integral as in the computation of $F_1$, we conclude that:
\begin{equation*}
    F_2 = \frac{1}{\pi} \Big(-\frac{a\pi}{2}-\frac{a \pi }{2} \Big) = -a = \frac{1}{2} F_1 -\frac{a}{2^1}.
\end{equation*}
\begin{lemma}
For $1 \leq j+1 \leq n-1$, the following recursion relation holds:
\begin{equation}
    F_{j+1}=\frac{F_j}{2}-\frac{a}{2^j}. 
    \label{eq:F_j_induction}
\end{equation}   
\end{lemma}

\begin{proof}
By the definition,
\begin{equation*}
F_{j+1} = \int_0^1 d \lambda_{j+1} p(\lambda_{j+1})\lambda_{j+1}\ln(\lambda_{j+1}).    
\end{equation*}
Using the explicit form of $p(\lambda_{j+1})$ and exchanging the order of integration, we obtain:
\begin{align*}
    F_{j+1} & = \frac{1}{\pi^{j+1}} \int_0^1 d \lambda_1 \ldots  \int_0^{1-\sum_{i=1}^{j-1} \lambda_i} d \lambda_{j} \int_0^{1-\sum_{i=1}^{j} \lambda_i}  d \lambda_{j+1} \Big( \prod_{k=1}^j \frac{1}{\sqrt{\lambda_k(1-\sum_{i=1}^k\lambda_{i})}} \Big) \frac{\lambda_{j+1} \ln(\lambda_{j+1})}{\sqrt{\lambda_{j+1}(1 - \sum_{i=1}^{j+1} \lambda_{i})}}\\
    & = \frac{1}{\pi^{j+1}} \int_0^1 d \lambda_1 \ldots  \int_0^{1-\sum_{i=1}^{j-1} \lambda_i} d \lambda_{j}\Big( \prod_{k=1}^j \frac{1}{\sqrt{\lambda_k(1-\sum_{i=1}^k\lambda_{i})}} \Big) \int_0^{1-s}  d \lambda_{j+1}  \frac{\lambda_{j+1} \ln(\lambda_{j+1})}{\sqrt{\lambda_{j+1}(1 - s-\lambda_{j+1} )}}\\
\end{align*}
where $s=\sum_{i=1}^{j} \lambda_i $. 
The innermost integral is precisely $A(s)$ defined in~\cref{eq:A_recursion_F_2}, yielding 
\begin{align*}
    F_{j+1} 
     = & \frac{1}{\pi^{j+1}} \int_0^1 d \lambda_1 \ldots  \int_0^{1-\sum_{i=1}^{j-1} \lambda_i} d \lambda_{j}\Big( \prod_{k=1}^j \frac{1}{\sqrt{\lambda_k(1-\sum_{i=1}^k\lambda_{i})} }\Big) \pi (1-s) \Big( -a  + \ln(1 - s) \frac{1}{2 }\Big) \\
    = &  \frac{1}{2}\frac{1}{\pi^{j}} \int_0^1 d \lambda_1 \ldots  \int_0^{1-\sum_{i=1}^{j-1} \lambda_i} d \lambda_{j} (1-\sum_{i=1}^{j} \lambda_i) \ln(1-\sum_{i=1}^{j} \lambda_i)\Big(  \prod_{k=1}^j \frac{1}{\sqrt{\lambda_k(1-\sum_{i=1}^k\lambda_{i})}} \Big) \\
    & -a \frac{1}{\pi^{j}} \int_0^1 d \lambda_1 \ldots  \int_0^{1-\sum_{i=1}^{j-1} \lambda_i} d \lambda_{j} (1-\sum_{i=1}^{j} \lambda_i) \Big(  \prod_{k=1}^j \frac{1}{\sqrt{\lambda_k(1-\sum_{i=1}^k\lambda_{i})} }\Big)  \\
    = &  \frac{F_{j}}{2}-a \frac{1}{\pi^{j}} \int_0^1 d \lambda_1 \ldots  \int_0^{1-\sum_{i=1}^{j-1} \lambda_i} d \lambda_{j} (1-\sum_{i=1}^{j} \lambda_i) \Big(  \prod_{k=1}^j \frac{1}{\sqrt{\lambda_k(1-\sum_{i=1}^k\lambda_{i})} }\Big)   \;.
\end{align*}
Then, we have:
\begin{align*}
    B \equiv &  \int_0^1 d \lambda_1 \ldots  \int_0^{1-\sum_{i=1}^{j-1} \lambda_i} d \lambda_{j} (1-\sum_{i=1}^{j} \lambda_i) \Big(  \prod_{k=1}^j \frac{1}{\sqrt{\lambda_k(1-\sum_{i=1}^k\lambda_{i})} }\Big) \\
    = &  \int_0^1 d \lambda_1 \ldots \int_0^{1-\sum_{i=1}^{j-2} \lambda_i} d \lambda_{j-1}   \Big(  \prod_{k=1}^{j-1} \frac{1}{\sqrt{\lambda_k(1-\sum_{i=1}^k\lambda_{i})} }\Big)  \int_0^{1-\sum_{i=1}^{j-1} \lambda_i} d \lambda_{j} \sqrt{  \frac{1-\sum_{i=1}^{j} \lambda_i}{\lambda_j }}\\ 
    = &  \int_0^1 d \lambda_1 \ldots \int_0^{1-\sum_{i=1}^{j-2} \lambda_i} d \lambda_{j-1}   \Big(  \prod_{k=1}^{j-1} \frac{1}{\sqrt{\lambda_k(1-\sum_{i=1}^k\lambda_{i})} }\Big) (1-\sum_{i=1}^{j-1}\lambda_i)\frac{\pi}{2}\\
    & \quad \vdots \\
    = & \Big(\frac{\pi}{2}\Big)^j\\ 
\end{align*}
using the recurring structure. 
Substituting this result yields~\cref{eq:F_j_induction}.
\end{proof}
\begin{lemma}
    For $1\leq j+1 \leq n-1$, we have
\begin{equation}
    F_{j+1} = - \frac{(j+1)a}{2^j}
    \label{eq:f_expression_formula}
\end{equation}
where $a = \ln(2)-\frac{1}{2}$.
\end{lemma}
\begin{proof} We prove the result by recursion.
We have already shown that
\begin{equation*}
F_{0+1} = -a = -\frac{a(1+0)}{2^0} \;. 
\end{equation*}
Assume that~\cref{eq:f_expression_formula} holds for $j+1<n-1$, that is $F_{j+1} = -\frac{(j+1) a}{2^j}$. 
Then, using~\cref{eq:F_j_induction}, we obtain:
\begin{align*}
    F_{j+2} & = \frac{F_{j+1}}{2} -\frac{a}{2^{j+1}} \\
    & = - \frac{(j+2) a}{2^{j+1}} \;.
\end{align*}
which completes the recursion. 
\end{proof}

\begin{lemma}
\begin{equation*}
    F_n=F_{n-1}
\end{equation*}
\end{lemma}
\begin{proof}
By definition,
\begin{equation*}
    F_{n} \equiv \langle \bm{\Lambda_{n}} \ln \bm{\Lambda_{n}} \rangle = \int_0^1 d\lambda_n p(\lambda_n) \lambda_n \ln \lambda_n \;.
\end{equation*}
In~\cref{annex:pdf_eigenvalues}, we calculated the~\gls{PDF} $p(\lambda_j)$ for $1\leq j\leq n-1$, but we did not give an explicit expression for $p(\lambda_n)$. 
We therefore compute it here.
We note that
\begin{align*}
    p(\lambda_n) \equiv & ~p(\bm{\Lambda_n} = \lambda_n  )\\
    = & \int_0^{\pi/2} d\varphi_1 \ldots \int_0^{\pi/2} d\varphi_{n-2} \int_0^{\pi/2} d\varphi_{n-1} ~ \delta(\lambda_n,(\sin\varphi_1 \sin\varphi_2 \ldots \sin \varphi_{n-1})^2) \\
    & p(\bm{\varPhi_1}=\varphi_1, \ldots, \bm{\varPhi_{n-2}}=\varphi_{n-2}, \bm{\varPhi_{n-1}}=\varphi_{n-1})\\
    = & \int_0^{\pi/2} d\varphi_1 p(\bm{\varPhi_1}=\varphi_1) \ldots \int_0^{\pi/2} d\varphi_{n-2} p(\bm{\varPhi_{n-2}}=\varphi_{n-2}) \int_0^{\pi/2} d\varphi_{n-1} ~ \delta(\lambda_n - (\sin\varphi_1 \sin\varphi_2 \ldots \sin \varphi_{n-1})^2) \\
    & p(\bm{\varPhi_{n-1}}=\varphi_{n-1}) \;,
\end{align*}
using the fact that $\bm{\varPhi_1}, \ldots, \bm{\varPhi_{n-2}}, \bm{\varPhi_{n-1}}$ are independent random variables to transform the second line into the third line. 
Here, $\delta$ denotes the Dirac delta distribution. 
Since $\bm{\varPhi_{n-1}}$ is a uniformly distributed random variable on $[0, \pi/2[$, its~\gls{PDF} is given by:
\begin{equation*}
    p(\varphi_{n-1}) =
\begin{cases} 
    \frac{2}{\pi} & \text{ if } \varphi_{n-1} \in [0,\frac{\pi}{2}[, \\
    0 & \text{ otherwise. }  
\end{cases}
\end{equation*} 
We define
\begin{align*}
    I \equiv & \int_0^{\pi/2} d\varphi_{n-1} ~ \delta(\lambda_n - (\sin\varphi_1 \sin\varphi_2 \ldots \sin \varphi_{n-1})^2)~ p(\bm{\varPhi_{n-1}}=\varphi_{n-1}) \\
    = & \int_0^{\pi/2} d\varphi_{n-1} ~ \delta(\lambda_n-(\sin\varphi_1 \sin\varphi_2 \ldots \cos( \frac{\pi}{2}-\varphi_{n-1}))^2) ~  \frac{2}{\pi} \;.
\end{align*}
Using the substitution $\tilde{\varphi}_{n-1} = \frac{\pi}{2}-\varphi_{n-1}$, we get 
$d \tilde{\varphi}_{n-1}= - d \varphi_{n-1}$,
and the integration limits transform as
\[
\varphi_{n-1} = 0 \implies \tilde{\varphi}_{n-1} = \frac{\pi}{2}, \qquad \varphi_{n-1} = \frac{\pi}{2} \implies \tilde{\varphi}_{n-1} = 0.
\]
Thus,
\begin{equation*}
    I  = \int_{0}^{\frac{\pi}{2}} d\tilde{\varphi}_{n-1} ~ \delta(\lambda_n -(\sin\varphi_1 \sin\varphi_2 \ldots \cos( \tilde{\varphi}_{n-1}))^2) ~  \frac{2}{\pi} \; .
\end{equation*}
Now, we define a random variable $\bm{\tilde{\varPhi}_{n-1}}$ to be uniformly distributed on $[0,\pi/2[$. We have:
\begin{equation*}
    I = \int_{0}^{\frac{\pi}{2}} d\tilde{\varphi}_{n-1} ~ \delta(\lambda_n-(\sin\varphi_1 \sin\varphi_2 \ldots \cos( \tilde{\varphi}_{n-1}))^2) ~  p(\bm{\tilde{\varPhi}_{n-1}}=\tilde{\varphi}_{n-1}) \;.
\end{equation*}
Thus, 
\begin{align*}
    p(\bm{\Lambda_{n}}=\lambda_{n})
    = & \int_0^\pi d\varphi_1 p(\bm{\varPhi_1}=\varphi_1) \ldots \int_0^\pi d\varphi_{n-2} p(\bm{\varPhi_{n-2}}=\varphi_{n-2}) \int_0^{2\pi} d\tilde{\varphi}_{n-1} ~ \delta(\lambda_n-(\sin\varphi_1 \sin\varphi_2 \ldots \cos \tilde{\varphi}_{n-1})^2) \\
    & p(\bm{\tilde{\varPhi}_{n-1}}=\tilde{\varphi}_{n-1}) \\
    = & p(\bm{\Lambda_{n-1}}=\lambda_{n}) \;.
\end{align*}
Therefore, the random variables $\bm{\Lambda_{n-1}}$ and $\bm{\Lambda_{n}}$ are identically distributed, so $F_{n}=F_{n-1}$.
\end{proof}
\begin{proposition}
    The expectation value of the von Neumann entropy is:
    \begin{equation*}
        \langle \bm{S_n} \rangle =(\ln(2)-\frac{1}{2}) \Big(4- \frac{1}{2^{n-3}} \Big) \;.
    \end{equation*}
\end{proposition}
\begin{proof}
Using the calculated expressions of $F_j$, we write:
\begin{align}
    \langle \bm{S_n} \rangle & = \sum_{j=1}^{n-1} \frac{aj}{2^{j-1}}+\frac{a(n-1)}{2^{n-2}} \nonumber \\
     & = a \Big( \sum_{j=0}^{n-2} \frac{j+1}{2^{j}} +\frac{(n-1)}{2^{n-2}} \Big) \nonumber \\
     &= a \Big( \sum_{j=0}^{n-2} \frac{j}{2^{j}}+\sum_{j=0}^{n-2} \frac{1}{2^{j}} +\frac{(n-1)}{2^{n-2}}\Big) \;. \label{eq:average_S_n}
\end{align}
Using the identity~\cite[p.~33, Eq.~2.26]{graham94}
\begin{equation*}
    \sum_{j=0}^{n} j x^j = \frac{x-(n+1)x^{n+1}+n x^{n+2}}{(1-x)^2},
\qquad x\neq1,
\end{equation*}
we obtain for $x=\frac{1}{2}$:
\begin{equation*}
\sum_{j=0}^{n-2} \frac{j}{2^j}
= 2-\frac{n}{2^{n-2}}.
\end{equation*}
The second term in parentheses~\cref{eq:average_S_n} is a geometric series with ratio $r=\frac{1}{2}< 1$, yielding:
\begin{equation*}
    \sum_{j=0}^{n-2} \frac{1}{2^{j}} = 2-\frac{1}{2^{n-2}}.
\end{equation*}
Substituting these expressions gives
\begin{align*}
    \langle \bm{S_n} \rangle & =    (\ln(2)-\frac{1}{2}) \Big(4- \frac{(n+1)}{2^{n-2}} + \frac{(n-1)}{2^{n-2}}\Big) \\
    & =  (\ln(2)-\frac{1}{2}) \Big(4- \frac{1}{2^{n-3}} \Big) \;.
\end{align*}

\end{proof}
\section{Generating a set of eigenvalues: Gaussian distribution around the maximally entangled point \label{annex:generating_the_eigenvalues_gaussian}}
\subsection{Point on the $n$-sphere corresponding to the maximally entangled state}
In this appendix, we determine the coordinates of the point on the $n$-sphere that leads to the maximally entangled state. 
We refer to this point as the \emph{maximally entangled point}.
A bipartite pure state $\rho \in \mathcal{H}_A \otimes \mathcal{H}_B$ is maximally entangled if the reduced density matrices of both subsystems are given by~\cite{Marco2016}:
\begin{equation*}
    \rho_A = \rho_B = \begin{bmatrix}
        \frac{1}{n} && \\
        & \ddots & \\
        && \frac{1}{n}
    \end{bmatrix}
\end{equation*}
where $n$ is the dimension of the smaller subsystem. 
In our case, we consider the bipartition $\mathcal{H}_A = \mathcal{H}_{1, \ldots, \ell} $, $\mathcal{H}_B =  \mathcal{H}_{\ell+1, \ldots, L}$ and $n=\min(d^{\ell},d^{L-\ell})$ where $\mathcal{H}_{1,\ldots, \ell}$ (respectively $\mathcal{H}_{\ell+1,\ldots, L}$) denotes the total Hilbert space associated with sites $1, \ldots, \ell$ (respectively $\ell +1, \ldots, L$). 
The von Neumann entropy of $\rho_A$ is then given by $S(\rho_A)=\ln (n)$~\cite[p.~505]{Nielsen_Chuang_2010}. 
By the bijection $\vec G$ defined in~\cref{eq:G_bijective}, there exists a unique point on the $n$-sphere whose squared Cartesian coordinates yield the eigenvalues of the reduced density matrix associated with the maximally entangled state on $\mathcal{H}_{1, \ldots, \ell} \otimes \mathcal{H}_{\ell+1, \ldots, L}$.
The eigenvalues of this state are $\lambda_i=\frac{1}{n}$ $\forall i \in \llbracket 1, n \rrbracket$. 
Consequently, the corresponding Cartesian coordinates on the $n$-sphere must satisfy $x_i=\frac{1}{\sqrt{n}}$. 
\begin{proposition}
    On $\mathcal{S}^{n}_+$ defined in~\cref{eq:definition_of_S_tilde_n}, the spherical coordinates corresponding to the Cartesian coordinates $x_i = \frac{1}{\sqrt{n}}$ are given by:
    \begin{equation*}
    \label{equ:induction_angle_values}
        \varphi_i=\arccos(\frac{1}{\sqrt{n-i+1}}), \qquad i = 1,\ldots,n-1. 
    \end{equation*}
\end{proposition}
\begin{proof}
On $\tilde{S}^n$, the spherical coordinates are related to the Cartesian coordinates by the transformation given in~\cref{eq:cartesian_to_spherical_coordinates}.
For $x_i=\frac{1}{\sqrt{n}}$, this relation yields
\begin{equation*}
    \varphi_i=\operatorname{atan2}(\sqrt{\frac{n-i}{n}},\frac{1}{\sqrt{n}})=\arctan(\sqrt{\frac{n-i}{n}}\sqrt{n})=\arccos(\frac{1}{\sqrt{n-i+1}}) \;.
\end{equation*}
\end{proof}

\subsection{Probability distribution of the eigenvalues}

Sampling eigenvalues in the vicinity of the maximally entangled point ensures that they remain close to those of a maximally entangled state and therefore exhibit a high degree of entanglement, characteristic of a volume law.
In this appendix, we describe how sets of eigenvalues are sampled around the maximally entangled point.
We choose the probability distribution of the points on the $n$-sphere to be a Gaussian distribution whose expectation value coincides with the maximally entangled point.
More precisely, we define the random vector $ (\bm{\varPhi_1} , \ldots, \bm{\varPhi_{n-1}} )^\text{T}$, and impose that each component $\bm{\varphi_i}$ is an independent Gaussian random variable distributed as:
\begin{equation*}
    \bm{\varPhi}_i \sim \mathcal{N}\!\left(
\arccos\!\left(\frac{1}{\sqrt{n-i+1}}\right),\, \sigma^2
\right),
\qquad i = 1,\ldots,n-1.
\end{equation*}
Here, $\sigma$ denotes the standard derivation and is chosen according to the desired proximity to the maximally entangled point.
\section{Warmup and sweeping procedure}
\label{annex:algorithm:sweeping}
In this appendix, we discuss the details of the algorithm to produce random~\glspl{MPS}.
Given a certain set of target Schmidt values $\{S^{[l]}_n\}_n$ for each bipartition, as well as a convergence threshold $\delta$, the algorithm generates~\gls{MPS} site matrices $\{A^{\sigma_l}\}_{l}$.
Optimally, these matrices represent an element of the $\sigma$\hyp ensemble, i.e., in each bipartition the set of Schmidt values $\{S^{[l]}_n\}_n$ of the state $\vert \psi[\{A^{\sigma_l}\}_{l}] \rangle$ is given by the target Schmidt values drawn according to the $\sigma$\hyp distribution.
However, due to the quantum marginal problem, this cannot generally be achieved by solving a set of linear equations generated by site matrices covering only a finite subset of the lattice sites.
Therefore, the first part of the algorithm is designed to generate a suitable initial guess state for a subsequent iterative optimization scheme, whose details are provided below.
The idea of the warmup procedure is to draw random site matrices, locally enforcing a given set of Schmidt values by solving~\cref{eq:svd:Sigma_Omega_B}
\begin{equation*}
    \Sigma^{[l]} \Omega^{[l]}M^{[l+1]} = U^{[l+1]}S^{[l+1]}W^{[l+1]} \; .
\end{equation*}

Here, we solve for the right\hyp canonical, vertically fused site matrices $M^{[l+1]}$ under the constraint that these matrices are drawn from the Haar\hyp random ensemble.
The problem is non\hyp linear in general, but this can be mitigated by choosing a certain gauge for the, at this point unknown, $U^{[l+1]}$ matrices.
However, in fixing this gauge, we choose a basis for the left bond basis in the current bipartition.
This is the crucial step, which may violate the representability condition, as we do not know, at this stage, whether the chosen basis is compatible with the set of Schmidt values for the bipartitions to the left of bond $l+1$.
We nevertheless found this procedure to provide suitable initial guess states for a subsequent algorithmic state and, in some cases, it even provided sufficiently precise representations of the $\sigma$\hyp ensemble.
However, in order to reliably obtain~\gls{MPS} representations we performed a subsequent iterative optimization detailed in the following.
The iterative optimization scheme is based on a simple~\gls{MPS} sweeping procedure.
We begin with a certain initial guess state represented by left\hyp canonical site matrices $\{A^{\sigma_l}\}_{sigma_l}$, which are obtained from the warmup procedure outlined above and in the main text.
The set of target Schmidt values for each bipartition between sites ${l,l+1}$ is denoted by $\{S^{[l]}\}$.
In a right\hyp to\hyp left sweep, the update procedure per site is then as follows:
\begin{enumerate}
    \item Stack the site matrices $M^{\sigma_l}$ row-wise to obtain \begin{equation}
    M^{[l]}
    =
    \left(
        \begin{array}{c}
            M^{\sigma_{l}=0} \\ 
            \vdots \\
            M^{\sigma_{l}=d-1}
        \end{array}
    \right) \; , 
    \end{equation}
    where $d$ again denotes the local Hilbert space dimension.
    \item Perform a~\gls{SVD} $M^{[l]} = U \tilde S^{[l]} V^{[l]}$
    \item Compute the distance $\epsilon_l = \Vert S^{[l]} - \tilde S^{[l]}\Vert$ measuring the deviation between the actual Schmidt values of the~\gls{MPS} and the targeted Schmidt values
    \item Vertically split 
    \begin{equation} V^{[l]} =
    \left(
        \begin{array}{c}
            V^{\sigma_{l}=0} \\ 
            \vdots \\
            V^{\sigma_{l}=d-1}
        \end{array}
    \right) \; , \end{equation} 
    and replace the current site tensors $V^{\sigma_l}\mapsfrom M^{\sigma_l}$. 
    \item Contract $U S^{[l]}$ into the next site tensors $A^{\sigma_{l-1}}U S^{[l]} \mapsfrom A^{\sigma_{l-1}}$.
\end{enumerate}
Once the site tensors $A^{\sigma_1}$ have been updated, the total distance to the target Schmidt value distributions $\epsilon = \sum_{l=1}^{L-1} \epsilon_l$ is evaluated.
If the convergence threshold is reached, i.e., $\epsilon < \delta$, the procedure can be terminated.
Otherwise, continue by sweeping reverting the sweep direction and repeat the above steps, now working on the left\hyp canonical $A^{\sigma_l}$ tensors.
\newpage

\end{document}